\documentclass[format=acmsmall, review=false, screen=false]{acmart}

\usepackage{booktabs} 

\acmJournal{TODS}
\acmVolume{1}
\acmNumber{1}
\acmArticle{1}
\acmYear{2018}
\acmMonth{1}
\acmPrice{15.00}
\copyrightyear{2018}

\setcopyright{acmlicensed}

\acmDOI{10.1145/3230634}

\usepackage{amssymb}
\usepackage{amsmath}

\usepackage{overpic}
\usepackage{graphicx}
\usepackage{booktabs}
\usepackage{verbatim}
\usepackage[caption=false,font=normalsize,labelfont=sf,textfont=sf]{subfig}
\usepackage[ruled,vlined,linesnumbered]{algorithm2e}

\SetAlFnt{\small}
\SetAlCapFnt{\small}
\SetAlCapNameFnt{\small}
\SetAlCapHSkip{0pt}
\IncMargin{-\parindent}

\newcommand{\eat}[1]{}

\newtheorem{thm}{Theorem}
\newtheorem{defn}{Definition}

\begin{document}
\title{K-Regret Queries Using Multiplicative Utility Functions}

\author{Jianzhong Qi}
\affiliation{%
  \institution{The University of Melbourne}
  \city{Melbourne}
  \state{Victoria}
  \postcode{3010}
  \country{Australia}
  }
\email{jianzhong.qi@unimelb.edu.au}

\author{Fei Zuo}
\affiliation{%
  \institution{The University of Melbourne}
  \city{Melbourne}
  \state{Victoria}
  \postcode{3010}
  \country{Australia}
  }
\email{fzuo@student.unimelb.edu.au}

\author{Hanan Samet}
\affiliation{%
  \institution{University of Maryland}
  \city{College Park}
  \state{Maryland}
  \postcode{20742}
  \country{USA}
  }
\email{jianzhong.qi@unimelb.edu.au}

\author{Jia Cheng Yao}
\affiliation{%
  \institution{The University of Melbourne}
  \city{Melbourne}
  \state{Victoria}
  \postcode{3010}
  \country{Australia}
  }
\email{yaoj1@student.unimelb.edu.au}

\begin{abstract}
The $k$-regret query aims to return a size-$k$ subset $\mathcal{S}$ of  a database $\mathcal{D}$ such that, for any query user that selects a data object from this size-$k$ subset $\mathcal{S}$ rather than from database $\mathcal{D}$, her regret ratio is minimized. The regret ratio here is modeled by the relative difference in the optimality between the locally optimal object in $\mathcal{S}$ and the globally optimal object in $\mathcal{D}$. The optimality of a data object in turn is modeled by a utility function of the query user. 
Unlike traditional top-$k$ queries, the $k$-regret query does not minimize the regret ratio for a specific utility function. Instead, 
it considers a family of infinite utility functions $\mathcal{F}$, and aims to 
find a size-$k$ subset that minimizes the maximum regret ratio of any utility function in $\mathcal{F}$. 

Studies on $k$-regret queries have focused on the family of additive utility functions, which 
have limitations in modeling individuals' preferences and decision making processes, especially for a common observation called the diminishing marginal rate of substitution (DMRS).
We introduce $k$-regret queries with multiplicative utility functions, which are more expressive in modeling the DMRS, to overcome those limitations.
We propose a query algorithm with bounded regret ratios. 
To showcase the applicability of the algorithm, we apply it to a special family of multiplicative utility functions, the Cobb-Douglas family of utility functions, and a closely related family of  utility functions, the Constant Elasticity of Substitution family of utility functions, both of which are frequently used utility functions in microeconomics. 
After a further study of the query properties, we propose a heuristic algorithm that produces even smaller regret ratios 
in practice. Extensive experiments on the proposed algorithms confirm that they consistently achieve small maximum regret ratios. 

\end{abstract}

%
%


%
%


\thanks{
Authors' address: 
J. Qi, F. Zuo, and J. C. Yao, 
The University of Melbourne, Australia,
Email: jianzhong.qi@unimelb.edu.au, \{fzuo, yaoj1\}@student.unimelb.edu.au;
H. Samet, 
University of Maryland, USA,
Email: hjs@cs.umd.edu.}

\maketitle


\section{Introduction}\label{sec:intro}

Ever growing data have produced various databases that are beyond any user's capability to explore them in full. 
For example, Amazon has a database of over 562 million products~\cite{amazon}; Booking.com has a database of over 1.7 million hotels~\cite{booking}. 
\emph{Window queries}~\cite{Aref97} (a.k.a. \emph{range queries}~\cite{Amir01,Ang90}), 
\emph{top-$k$ queries}~\cite{DBLP:conf/icde/ChenCCT15,DBLP:conf/cikm/HuangWQZCH11,Ilyas:2008:STK:1391729.1391730,soliman2008probabilistic,DBLP:journals/tkde/WuYCJ12}, and \emph{skyline queries}~\cite{borzsony2001skyline,DBLP:conf/icde/KhalefaML08,DBLP:conf/cikm/KhalefaML10,papadias2003optimal,tan2001efficient} have been used traditionally to 
produce a representative subset $\mathcal{S}$ when a database $\mathcal{D}$ is too large. 
Window queries return data objects with attribute values falling in certain ranges, which may not represent the distribution of the full database. 
Top-$k$ and skyline queries, on the other hand, suffer  by either requiring a predefined \emph{utility function} to 
model user preferences over the data objects,  or returning an unbounded number of data objects. 
Recent studies~\cite{kessler2015k,nanongkai2010regret,zeighami2016minimizing} aim to overcome these limitations by a new type of query,  
the \emph{$k$-regret query}, which returns 
a size-$k$ subset $\mathcal{S} \subseteq \mathcal{D}$ that minimizes the 
\emph{maximum regret ratio} of any query user. 
The set $\mathcal{S}$ is called the \emph{regret-minimizing set}.
The concept of \emph{regret} comes from microeconomics~\cite{ligett2011beating}.
Intuitively, if a query user had selected the locally optimal object in $\mathcal{S}$, and were later 
shown the globally optimal object in $\mathcal{D}$, then the query user may have some regret.
A $k$-regret query uses the \emph{regret ratio}
to model how regretful the query user may be, which 
is the relative difference in the \emph{optimality} 
between the locally optimal object and the globally optimal object.
Here, the optimality is computed by a utility function.  
The $k$-regret query considers a family of infinite utility functions such as the family of linear 
summation functions. It aims to find the subset $\mathcal{S}$
that minimizes the maximum regret ratio for any utility function in such a function family.

\begin{table}
\centering
\caption{A  Computer Database}\label{tbl:example}
\BlankLine\BlankLine
\begin{tabular}{cccccccc}
\toprule
\multicolumn{1}{c}{Computer} & CPU ($p_i.c_1$) & Brand recognition ($p_i.c_2$) & $f_1(p_i)$ & $f_2(p_i)$ & $f_3(p_i)$ & $f_4(p_i)$ \\
\midrule
$p_1$ & 2.3 & 80 & 41.15	& 3.08	& \textbf{13.56} &	2.38\\
$p_2$ & 1.7 & \textbf{90} & \textbf{45.85}	& 2.58 &	12.37 &	1.77\\
$p_3$ & 2.8 & 50 & 26.40	& 3.27 &	11.83 &	2.88\\
$p_4$ & 2.1 & 55 & 28.55	& 2.63 &	10.75 &	2.17\\
$p_5$ & 2.1 & 50 & 26.05	& 2.58 &	10.25 &	2.17\\
$p_6$ & \textbf{3.0} & 55 & 29.00	& \textbf{3.52}	& 12.85	& \textbf{3.09}\\
\bottomrule
\end{tabular}
\end{table}

To illustrate the $k$-regret query, 
consider an online computer shop with a database $\mathcal{D}$ of  computers as shown in Table~\ref{tbl:example}.
There are six computers: $\mathcal{D} = \{p_1, p_2, ..., p_6\}$. Every computer $p_i$ has two attributes: CPU clock speed and brand recognition, 
denoted as $p_i.c_1$ and $p_i.c_2$, respectively. 
Here, brand recognition represents how well a brand is recognized by the customers. A larger value means that the brand is better recognized. 
Since the database may be too large to be shown in its entirety,
the shop considers showing only a size-$k$ subset $\mathcal{S} \subseteq \mathcal{D}$ in the front page as a recommendation. 
Such a subset may be $\mathcal{S} = \{p_1, p_3, p_5\}$ (i.e., $k = 3$).
Suppose that there is a customer whose preference can be expressed as a utility function $f_1(p_i) = 0.5 \cdot p_i.c_1 + 0.5 \cdot p_i.c_2$.
The customer may purchase $p_1$ from the recommended subset $\mathcal{S}$ since $p_1$ has the largest utility value: $f_1(p_1) = 
0.5 \cdot p_1.c_1 + 0.5 \cdot p_1.c_2 = 0.5\times 2.3+0.5\times 80 = 41.15 > f_1(p_3) = 26.40 > f_1(p_5) = 26.05$.  Note that  
another computer $p_2 \in \mathcal{D}$ exists with an even larger utility value $f_1(p_2) = 0.5\times 1.7+0.5\times 90=45.85$. If the customer  
later sees $p_2$, she may have some regret. Her regret ratio is computed as $\displaystyle \frac{f_1(p_2) - f_1(p_1)}{f_1(p_2)} \approx 10.25\%$.
For another customer with a different utility function $f_2(p_i) = 0.99 \cdot p_i.c_1 + 0.01 \cdot p_i.c_2$, the computer in $\mathcal{S}$ that 
best suits her preference is $p_3$: $f_2(p_3) = 0.99 \times 2.8+0.01\times 50 \approx 3.27 > f_2(p_1) \approx 3.08 > f_2(p_5) \approx 2.58$.
For this customer, the globally optimal computer in $\mathcal{D}$ is $p_6$: $f_2(p_6) = 0.99 \times 3.0+0.01\times 55 = 3.52$. 
If the customer purchases $p_3$, her regret ratio will
be $\displaystyle \frac{f_2(p_6) - f_2(p_3)}{f_2(p_6)} \approx 7.05\%$.
Since customers have different preferences and different utility functions, 
different data objects are needed to minimize their regret ratios. 
It is unlikely that a size-$k$ subset can satisfy all the customers. 
The  $k$-regret query addresses this limitation by finding a subset $\mathcal{S}$ that 
minimizes the \emph{maximum regret ratio} for a family of infinite utility functions.

Existing studies on $k$-regret queries have focused on the families of \emph{additive utility functions} (AUFs) where the overall utility of an object is
computed as the sum of the utility in each attribute of the object. The linear summation functions $f_1$ and $f_2$ are examples. 
They can be written in a more general form: 
$$f(p_i) = \sum_{j=1}^{d}\alpha_j \cdot p_i.c_j,$$ 
where $d$ denotes the number
of attributes, and $\alpha_j \in [0,1]$ is the \emph{weight} of attribute~$j$. 
Studies~\cite{kessler2015k,nanongkai2010regret} have shown that the maximum regret ratio 
of the $k$-regret query with AUFs can be bounded. 

A significant limitation of AUFs, however, is that the overall utility of an object 
 always increases at the same rate as the utility in an attribute increases. 
For an AUF $f(p_i) = \sum_{j=1}^{d}\alpha_j \cdot p_i.c_j$, the value of $f(p_i)$ 
always increases by $\alpha_j$ units for every unit of increase in attribute $j$.
A large increase in the value of an attribute may cause a dramatic change in the value of 
the overall utility. 
Objects with the maximum values in 
certain attributes tend to be favored by AUFs, especially when the value ranges vary for different attributes.
Consider again the example above. Utility function  $f_1$ favors  
$p_2$ which has the maximum value in $c_2$ but also the minimum value in $c_1$. 
The two attributes $c_1$ and $c_2$, however, have the same weight (0.5) in the utility function,  
 indicating that the user has the same preference towards the two attributes. 
The object favored by $f_1$ does not suit this preference.  Thus,
utility functions like $f_1$ 
do not model individuals' preferences and decision making processes well. 

Intuitively, as the value of an attribute gets larger, adding an extra unit to its value should 
contribute a smaller increment to the overall utility. 
For example, adding a 4GB RAM to an old home computer with a 512MB RAM would make a major  
difference in the user experience; adding a 4GB RAM to a server with 256GB RAM would probably go 
without notice. 
This is in fact a common observation in individuals' decision making process called  
the \emph{diminishing marginal rate of substitution} (DMRS)~\cite{cass1965optimum,diamond1965national,uzawa1962production,varian1992microeconomic}.
The DMRS refers to the principle that, as the utility in an attribute $j$ gets larger, 
the extent to which this utility can make up (substitute) for the utility in any other attribute $j'$ decreases (diminishes). 
Thus, as the utility of attribute $j$ gets larger, 
the increment of the overall utility when adding an extra unit to attribute $j$ decreases.

To overcome the limitation of AUFs, 
we introduce a new type of $k$-regret queries with utility functions  that are more expressive in modeling the DMRS 
-- the $k$-regret query with \emph{multiplicative utility functions} (MUFs).  
An MUF computes  the overall utility of an object as the product of the utility in each attribute: 
$$f(p_i) = \prod_{j=1}^{d} p_i.c_j^{\alpha_j}$$
An MUF helps deal with exponential-like utility functions and is more expressive in modeling the DMRS for the following reason. Since the attribute value $p_i.c_j$ 
has been raised to the power equal to the weight $\alpha_j$, as $p_i.c_j$ gets larger, the increment of $f(p_i)$ when adding an extra unit to $p_i.c_j$ decreases. 
This is because the function $g(p_i.c_j) = (p_i.c_j+1)^{\alpha_j} - p_i.c_j^{\alpha_j}$ is monotonically decreasing (i.e., $\forall \alpha_j \in [0,1]: g'(p_i.c_j) \le 0$).
An MUF models user preferences towards different attributes better. 
For example, an MUF $f_3(p_i) = p_i.c_1^{0.5} \cdot p_i.c_2^{0.5}$ has the same weight in the two attributes.
It favors $p_1$ in the example above, where $f_3(p_1) =  2.3^{0.5} \times 80^{0.5} \approx 13.56$ is 
larger than the function value of any other object. 
Object $p_1$ does not have the maximum value in either attribute but is reasonably good in both attributes. 
It suits the user preference.  
The MUF is also more robust to large value changes in an attribute, since a weight with a value between 0 and 1 is in the power computation. 
The varying value ranges caused by the different types of attributes 
can be easily handled by an MUF, enabling the use of the attribute values  in their natural form 
without normalization.

A related family of utility functions called the \emph{Constant Elasticity of Substitution (CES) functions} 
studied earlier~\cite{kessler2015k} can also model the DMRS.
CES functions are not MUFs, and they have limitations as discussed below.
A CES utility function $f(p_i)$ has the form of $f(p_i) = (\sum^d_{j=1}\alpha_j \cdot p_i.c_j^b)^{\frac{1}{b}}$ where $b$ is a system parameter. 
Raising $p_i.c_j$ to the power of $b$ allows a CES utility function to model the DMRS, making the function a popular utility function~\cite{uzawa1962production,varian1992microeconomic,vilcu2011some}. 
However, parameter $b$ is a constant across all attributes. This leads to less flexibility in modeling different diminishing marginal rates of substitution over different attributes~\cite{10.2307/1805230}. 
As studies in economics show~\cite{10.2307/1936009,miller2008}, finding a suitable value of $b$ to fit a CES  utility function to real data can be difficult and is sensitive to data construction. 
In comparison, an MUF raises $p_i.c_j$ to the power of $\alpha_j$ which can have different values for different attributes.
This allows different diminishing marginal rates of substitution over different attributes  and can be easier to fit user preferences.

The higher expressive power of MUFs brings challenges in bounding the maximum regret ratio for them.
It is difficult to tightly bound the product of a series of exponential expressions. 
To the best of our knowledge, so far, no existing bound 
has been obtained for  the $k$-regret query with MUFs.
We overcome the challenges with a novel algorithm that we call \emph{MinVar}. 
MinVar is an adaption of  the \emph{CUBE}~\cite{nanongkai2010regret} and the \emph{MinWidth}~\cite{kessler2015k} algorithms, 
which are $k$-regret query algorithms for AUFs.
The MinVar algorithm partitions the data space into multiple buckets.
Together the buckets enclose all the data objects, and one object in each bucket is returned to form the answer set $\mathcal{S}$. 
For any utility function, the corresponding optimal object $p^*$ must be in some bucket. 
There is an object $s^*$ in this bucket that has been returned in $\mathcal{S}$. 
The distance between $p^*$ and $s^*$ is bounded by the
range of attribute values spanned by the bucket in each attribute, 
which further bounds the maximum regret ratio. 
Our contribution in MinVar is a novel space partitioning strategy based on 
data distribution, which produces tighter buckets in practice.  
More importantly, we make theoretical contributions by showing 
that the MinVar algorithm can obtain a maximum regret ratio bounded between $\displaystyle \Omega(\frac{1}{k^2})$ and $\displaystyle O(\ln(1+\frac{1}{k^{\frac{1}{d-1}}}))$ for $k$-regret queries with MUFs, where $d$ denotes the number of data attributes. 
To showcase the applicability of the MinVar algorithm in real world scenarios, we apply it on $k$-regret queries with a special family of MUFs, the  \emph{Cobb-Douglas functions}, which is used extensively in economics studies for modeling the DMRS~\cite{10.2307/1811556,DIAMOND19801,vilcu2011geometric}. 
As a by-product, we derive a new upper bound $\displaystyle O(\frac{1}{k^{\frac{1}{d-1}}})$ on the maximum regret ratio for $k$-regret queries with 
CES functions~\cite{vilcu2011geometric}. This upper bound  is tighter than a previously obtained bound~\cite{kessler2015k}, while it also applies to 
the MinWidth algorithm~\cite{kessler2015k} proposed for $k$-regret queries with CES functions.

MinVar aims to \emph{bound} the maximum regret ratios rather than to \emph{minimize} them. 
Its bucket-based answer object selection strategy is conservative. 
It works well when there are a large number of data objects lying in different buckets, each of which is optimal for a different utility function. 
However, in real data sets, many of the objects may not be optimal for any utility function. Some of the buckets created by MinVar  
may not contain any data objects that are optimal for any utility function. Returning objects in those buckets does not 
contribute to lowering the maximum regret ratios. 
We will show that, for MUFs, the set $\mathcal{P}$ of all the \emph{skyline points} (objects)~\cite{borzsony2001skyline} 
in a database $\mathcal{D}$ minimizes the maximum regret ratio, which is 0. 
This is because any non-skyline point $p_i$ must be dominated by at least one skyline point $p_j$, 
and hence its utility $f(p_i)$ does not exceed the utility $f(p_j)$ for any MUF $f$.
If $|\mathcal{P}| \le k$, the entire set of $\mathcal{P}$ should be returned as the query answer set. 
Otherwise, we need to select $k$ skyline points to form a size-$k$ answer set.
We use the regret ratio to guide the selection of skyline points such that the difference in the utilities of the selected and unselected points
is  minimized. This leads to a heuristic algorithm named \emph{MaxDif} that produces even smaller maximum regret ratios  in practice.

To summarize, our paper makes the following contributions:
\begin{itemize}
\item We introduce a novel type of  $k$-regret queries -- $k$-regret queries with multiplicative utility functions, which are 
more expressive in modeling the diminishing marginal rate of substitution in making decisions.

\item  We propose an algorithm named MinVar to process the query and to bound the maximum regret ratios. 
Based on this algorithm, we obtain bounds of the maximum regret ratio for the $k$-regret query with 
multiplicative utility functions. 

\item The MinVar algorithm is deigned for multiplicative utility functions but it can also be applied to non-multiplicative utility functions. 
We showcase such applicability via two families of utility functions used in economic studies: 
(i) the Cobb-Douglas family of utility functions, which is a special type of multiplicative utility functions that has not been studied before in the context of  $k$-regret queries, and (ii) 
the CES family of utility functions, which is a family of non-multiplicative utility functions but is closely related to the Cobb-Douglas family of utility functions.
As a by-product, we derive   
an upper bound on the maximum regret ratio for $k$-regret queries with CES utility functions
that is tighter than an existing bound~\cite{kessler2015k} under the case where the function parameter $b\in(0,1)$. 
This bound applies to our MinVar algorithm as well as the MinWidth algorithm~\cite{kessler2015k} proposed for $k$-regret queries with CES utility functions.

\item 
Since MinVar aims to bound the maximum regret ratios rather than to minimize them, 
we further propose a heuristic algorithm named MaxDif that computes a size-$k$ subset of skyline points
to minimize the maximum regret ratios. 

\item We perform extensive experiments using both real and synthetic data to verify the effectiveness and efficiency 
of the proposed algorithms. The results show that the maximum regret ratios obtained by the proposed algorithms  are consistently small. 
Meanwhile, the proposed algorithms are more efficient than the baseline algorithm MaxDom~\cite{lin2007selecting}, which is a heuristic algorithm that computes 
top-$k$ representative skyline points.
\end{itemize}

The rest of the paper is organized as follows. 
Section~\ref{sec:relatedwork} reviews related work. 
Section~\ref{sec:preliminaries} presents the basic concepts.
Sections~\ref{sec:minvar} and~\ref{sec:bounds} describe the MinVar algorithm 
and derive bounds on the maximum regret ratio for $k$-regret queries with MUFs, respectively. 
Section~\ref{sec:casestudy}  showcases the applicability of MinVar to both 
MUFs and non-MUFs.
Section~\ref{sec:maxdif} presents the heuristic algorithm MaxDif.
Section~\ref{sec:exp} examines the results of our experiments while Section~\ref{sec:conclusions} concludes the paper.

\section{Related Work}\label{sec:relatedwork}
We review two queries: skyline and $k$-regret. 

\textbf{Skyline queries.}
The skyline query~\cite{borzsony2001skyline} is an earlier attempt to generate a representative subset $\mathcal{S}$ of a database $\mathcal{D}$ without 
specifying any utility functions. 
This query is defined based on the \emph{dominance} relationship. 
It considers a database $\mathcal{D}$ of $d$-dimensional points ($d \in \mathbb{N}_+$). 
Let $p_i$ and $p_j$ be two points in $\mathcal{D}$. 
Point  $p_i$ is said to \emph{dominate} point $p_j$ if and only if $ \forall l \in [1..d], p_i.c_l \ge p_j.c_l \wedge \exists l \in [1..d], p_i.c_l > p_j.c_l$, 
where $p_i.c_l$ ($p_j.c_l$) denotes the coordinate of $p_i$ ($p_j$) in dimension $l$. Here, the ``$\ge$'' and ``$>$'' operators represent the preference relationship. 
A point with a larger coordinate in dimension $l$ is  
preferable in that dimension. 
The skyline query returns the subset $\mathcal{S} \subseteq \mathcal{D}$ where each point is \emph{not} dominated by any other point in  
$\mathcal{D}$. 
It is interesting to observe that the attributes in the domain over 
which the skyline query is executed do not have to be 
spatial~\cite{Same04}, as is the case when they are embedded in a 
spatial database (e.g.,~\cite{Espe02,Same03,Same87a}), 
nor is there a requirement for a distance
function to exist between the objects (e.g., Euclidean or 
Hausdorff~{\cite{Nuta11}).

The skyline query can be answered by a two-layer nested loop over the points in $\mathcal{D}$ and another 
layer of loop over the $d$ dimensions to filter out the points dominated. 
The remaining points are  \emph{skyline points} which are the query answer. More efficient algorithms have been 
 proposed in the literature~\cite{papadias2003optimal,tan2001efficient} but are not the focus of our study. 

While the skyline query does not require a utility function, it suffers in lacking control over the size of the answer set. 
In the worst case, the entire database may be returned. 
Studies have tried to overcome this limitation by combining the skyline query with the top-$k$ query. 
For example, Xia et al.~\cite{xia2008skylining} introduce the \emph{$\varepsilon$-skyline} which adds a  \emph{weight} to each dimension of 
the data points to reflect user preference towards the dimension. 
The weights create a built-in rank for the points which can be used to answer  the \emph{top-$k$ skyline query}.
Chan et al.~\cite{chan2006high} rank the points by the \emph{skyline frequency}, i.e., how frequently a point appears as a skyline point when different numbers 
of dimensions are considered.  
A few other studies extract a representative subset of the skyline points. 
Lin et al.~\cite{lin2007selecting} propose to return the $k$ points that together dominate the most non-skyline points as the \emph{$k$ most representative skyline subset}. 
Tao et al.~\cite{tao2009distance} select 
\emph{$k$ representative skyline points} based on the distance between the skyline points instead. These studies 
bound the size of the answer set. However, they do not bound the maximum regret ratio of the set.

\textbf{$K$-regret queries.}
Nanongkai et al.~\cite{nanongkai2010regret} introduce the concept of \emph{regret minimization} to top-$k$ query processing
and propose the \emph{$k$-regret query}. 
This query  does not require query users to specify their utility functions. Instead, 
it considers a family of infinite utility functions, and finds the subset $\mathcal{S}$ that minimizes the \emph{maximum regret ratio} 
of the entire family of utility functions. 
Nanongkai et al. propose the \emph{CUBE} algorithm to process the $k$-regret query with the family of linear summation utility functions, 
i.e., each utility function $f$ is in the form of $f(p_i) = \sum_{j = 1}^{d} \alpha_j \cdot p_i.c_j$ where $\alpha_j$ denotes the \emph{weight} of  dimension $j$. 
The CUBE algorithm is efficient, but the maximum regret ratio it obtains is quite large in practice. To obtain a smaller maximum regret ratio,
in a different paper~\cite{nanongkai2012interactive}, Nanongkai et al. propose an interactive algorithm where query users are involved in
guiding the search for answers with smaller regret ratios. 
Peng and Wong~\cite{peng2014geometry} advance  the $k$-regret query studies 
by utilizing geometric properties to improve the query efficiency. 
Asudeh et al.~\cite{Asudeh:2017:ECR:3035918.3035932} use the \emph{convex hull}  
to find the data points that minimize the maximum regret ratio for linear summation utility functions. 
They propose an algorithm that can approximate the maximum regret ratio to within a user given threshold.
Cao et al.~\cite{cao_et_al:LIPIcs:2017:7056} and Chester et al.~\cite{chester2014computing} 
also consider linear summation utility functions but compute the \emph{$k$-regret minimizing sets}, which is NP-hard.

Kessler Faulkner et al.~\cite{kessler2015k} build on top of  CUBE and propose three algorithms, \emph{MinWidth}, \emph{AreaGreedy}, and \emph{Angle}. 
These three algorithms can process $k$-regret queries with 
the ``concave'', ``convex'', and CES utility functions. Nevertheless, the  ``concave'' and ``convex'' utility functions
considered have focused on \emph{additive forms} (See Braziunas and Boutilier~\cite{braziunas2007minimax} and Keeney and Raiffa~\cite{keeney1993decisions} for 
more details on \emph{additive utilities} and \emph{additive independence}). They are summations over a set of concave and convex functions. 
The CES utility functions also sum over a set of terms. 
In this paper, we introduce the use of the family of multiplicative utility functions to overcome 
the linearity limitation of the additive utility functions.  
We present an algorithm that can produce answers with bounded
maximum regret ratios for $k$-regret queries with multiplicative utility functions. 
As a by-product, we also derive a new upper bound on the maximum 
regret ratio for $k$-regret queries with CES utility functions which is tighter than a previously obtained upper bound~\cite{kessler2015k}, while the bound also applies 
to the MinWidth algorithm. 

Zeighami and Wong~\cite{zeighami2016minimizing} propose to compute the \emph{average regret ratio}. They do not 
assume any particular type of utility functions, but use sampling to obtain a few utility functions for the computation. 
This study is less relevant to our work and is not discussed further.

Note that Chester et al.~\cite{chester2014computing} have used the term \emph{$k$-regret minimizing set} to denote a subset  $\mathcal{S}$ of size $r$ 
that minimizes the \emph{maximum $k$-regret ratio}, where the regret is measured by the utility difference between  
the optimal point in $\mathcal{S}$ and the $k^{th}$ optimal point in the database $\mathcal{D}$.
We use the term \emph{$k$-regret query} following the closest related work~\cite{kessler2015k} 
to denote a query that finds the \emph{regret-minimizing set}, which is a subset  $\mathcal{S}$ of size $k$ that minimizes 
the \emph{maximum regret ratio}, where the regret is measured by the utility difference between the optimal points in $\mathcal{S}$ and $\mathcal{D}$.

The concept of regret has also been used in a classic problem in operations research -- the \emph{multi-armed bandit} (MAB) problem~\cite{10.2307/3689689}.
The MAB problem assumes $N$ \emph{arms} each associated with an unknown reward distribution. During a multi-round process, in each round, an agent chooses an arm and collects a reward generated by the corresponding reward distribution. 
Let the reward collected at round $i$ be $r_i$ and the largest expected reward of any arm be $\mu^*$.
The regret after $T$ rounds is $ T\mu^* - \sum_{i=1}^{T} \mathbb{E}(r_i)$. 
The key question in the MAB problem is how to balance the exploitation on the arm with the largest expected reward observed so far (to maximize $r_i$ for the current round) and the exploration  to find the arm with the globally largest expected reward  $\mu^*$ (to maximize $r_i$ for future rounds). This is less relevant to our problem, and we will not discuss this question further. 

\section{Preliminaries}\label{sec:preliminaries}

\begin{table}
\centering
\caption{Frequently Used Symbols}\label{tab:symbols}
\BlankLine\BlankLine
\begin{tabular}{cl}

\toprule
Symbol & Description \\
\midrule
$\mathcal{D}$ & A database\\
$n$ & The cardinality of $\mathcal{D}$\\
$d$ & The dimensionality of $\mathcal{D}$\\
$k$ & The $k$-regret query parameter\\
$\mathcal{S}$ & A size-$k$ subset selected from $\mathcal{D}$\\
$\mathcal{P}$ & The set of all the skyline points in $\mathcal{D}$\\
$p_i$& A point in $\mathcal{D}$\\
$p_i.c_j$& The coordinate value of $p_i$ in dimension $j$\\
$t$ & The number of intervals into which the data \\
       &   domain is partitioned in a dimension\\
\bottomrule
\end{tabular}
\end{table}

We present basic concepts and a problem definition in this section. The symbols frequently 
used in the discussion are  summarized in Table~\ref{tab:symbols}.

We consider a database $\mathcal{D}$ of $n$ data objects. Every data object $p_i \in \mathcal{D}$ 
is a $d$-dimensional point in $\mathbb{R}^{d}_+$, where $d$ is a positive integer and the coordinate values of the 
points are all positive numbers.
We use $p_i.c_j$ to denote the coordinate value of $p_i$ in dimension $j$. This coordinate value represents  
the \emph{utility} of $p_i$ in dimension $j$. A larger coordinate value denotes a larger utility 
and is preferable. 
A query parameter $k$  is given. It specifies the size of the 
answer set $\mathcal{S} \  (\mathcal{S} \subseteq \mathcal{D})$ to be returned. We assume $d \le k \le n$.

\emph{Gain.} Let $f: \mathcal{D} \rightarrow \mathbb{R_+}$ be a function that models the utility of a data object, i.e., 
how preferable the data object is by a query user.  
The \emph{gain} of a query user over a set $\mathcal{S}$, denoted by $gain(\mathcal{S},f)$, is  the maximum utility of any object in $\mathcal{S}$, i.e., 
\begin{equation}
gain(\mathcal{S},f)=\max_{p_i \in \mathcal{S}}f(p_i)
\end{equation}
Continuing with the example shown in Table~\ref{tbl:example}, if $\mathcal{S} = \{p_1, p_3, p_5\}$ and $f_3(p_i) = p_i.c_1^{0.5} \cdot p_i.c_2^{0.5}$,
\begin{align*}
 gain(\mathcal{S},f_3) &= \max_{p_i \in \mathcal{S}}f_3(p_i) =  f_3(p_1) =  2.3^{0.5} \times 80^{0.5} \approx 13.56
\end{align*}

\emph{Regret.} For a subset $\mathcal{S}$ of $\mathcal{D}$, the gain over $\mathcal{S}$ may be smaller than that of $\mathcal{D}$.
The difference between $gain(\mathcal{D},f)$ and $gain(\mathcal{S}, f)$ 
is the \emph{regret} of a query user 
if she selects the locally optimal object from $\mathcal{S}$ and is later shown the globally optimal object in $\mathcal{D}$,
 denoted by $regret_{\mathcal{D}}(\mathcal{S}, f)$: 
\begin{equation}
regret_{\mathcal{D}}(\mathcal{S},f) = gain(\mathcal{D},f) - gain(\mathcal{S}, f)
\end{equation}

\emph{Regret ratio.} 
The \emph{regret ratio}, $r\_ratio_{\mathcal{D}}(\mathcal{S}, f)$, is a relative measure of the regret. 
It is computed as the regret $regret_{\mathcal{D}}(\mathcal{S},f)$ over the gain $gain(\mathcal{D},f)$, i.e., 
\begin{equation}
 r\_ratio_\mathcal{D}(\mathcal{S},f) =\frac{regret_\mathcal{D}(\mathcal{S},f)}{gain(\mathcal{D},f)}
=\frac{max_{p_i\in \mathcal{D}}f(p_i)-max_{p_j\in \mathcal{S}}f(p_j)}{max_{p_i \in \mathcal{D}}f(p_i)}
\end{equation}

Continuing with the example in Table~\ref{tbl:example}, 
given $\mathcal{S} = \{p_1, p_3, p_5\}$ and $f_3(p_i) = p_i.c_1^{0.5} \cdot p_i.c_2^{0.5}$,
$gain(\mathcal{S},f_3) = gain(\mathcal{D},f_3)  =  f_3(p_1) \approx 13.56$. We have $regret_{\mathcal{D}}(\mathcal{S},f_3) = 0$ and  
$r\_ratio_\mathcal{D}(\mathcal{S},f_3) = 0\%$. Given a different utility function $f_4(p_i)= p_i.c_1^{0.99} \cdot p_i.c_2^{0.01}$, 
$gain(\mathcal{S},f_4) = f_4(p_3) \approx 2.88$ and $gain(\mathcal{D},f_4)  =  f_4(p_6)  \approx 3.09$. Then, $regret_{\mathcal{D}}(\mathcal{S},f_4) \approx 0.21$  and 
$r\_ratio_\mathcal{D}(\mathcal{S},f_4) \approx \frac{0.21}{3.09} \approx 6.80\%$.

\emph{Maximum regret ratio.} 
Given a set $\mathcal{S}$ and a family of utility functions $\mathcal{F}$, 
the \textit{maximum regret ratio} formulates how regretful a query user can be if her utility function is in $\mathcal{F}$.
It is the supremum of the regret ratio of a query user with any utility function in $\mathcal{F}$, i.e., 
\begin{equation}
\displaystyle mr\_ratio_\mathcal{D}(\mathcal{S},\mathcal{F})=\sup_{f\in\mathcal{F}}r\_ratio_\mathcal{D}(\mathcal{S},f)
 =\sup_{f\in\mathcal{F}}\frac{max_{p_i\in \mathcal{D}}f(p_i)-max_{p_j\in \mathcal{S}}f(p_j)}{max_{p_i \in \mathcal{D}}f(p_i)}
\end{equation}
Here, the supremum is used instead of the maximum because we consider an infinite set $\mathcal{F}$. 

Continuing with the example above, if $\mathcal{F} = \{f_3, f_4\}$, 
$$mr\_ratio_\mathcal{D}(\mathcal{S},\mathcal{F}) = \max \{0\%, 6.80\%\} = 6.80\%$$

The \emph{$k$-regret} query aims to return  the size-$k$ subset $\mathcal{S} \subseteq \mathcal{D}$ that \emph{minimizes} the 
maximum regret ratio for a family of utility functions.

\begin{defn}[$K$-Regret Query]
Given a family of utility functions $\mathcal{F}$, the \emph{$k$-regret} query returns a size-$k$ subset $\mathcal{S} \subseteq \mathcal{D}$, 
such that the maximum regret ratio over $\mathcal{S}$ is smaller than or equal to that over any other size-$k$ subset $\mathcal{S}' \subseteq \mathcal{D}$.
Formally,
\begin{align*}
 \forall \mathcal{S}' \subseteq \mathcal{D} \cap |\mathcal{S}'| = k: 
 mr\_ratio_\mathcal{D}(\mathcal{S},\mathcal{F}) \le mr\_ratio_\mathcal{D}(\mathcal{S}',\mathcal{F})
\end{align*}
\end{defn}

Specific utility functions are not always available because the query users are not usually known in advance
and their utility functions may not be specified precisely. 
The $k$-regret query does not require any specific utility functions to be given. Instead, the query considers 
a family of infinite functions such as the family of linear functions~\cite{nanongkai2010regret}, i.e., $f(p_i) = \sum_{j=1}^{d} \alpha_j \cdot p_{i}.c_{j}$ where 
$\alpha_j$ is the \emph{weight} of  dimension $j$. 
The $k$-regret query minimizes the 
maximum regret ratio of any utility function in such a family of utility functions, without knowing the value of the $\alpha_j$'s.

Our contribution to the study of $k$-regret queries is the incorporation of a family of \emph{multiplicative utility functions} (MUFs). 
\begin{defn}[Multiplicative Utility Function]
A \emph{multiplicative utility function (MUF)} $f$ is defined to be a utility function of the following form:
$$f(p_i) = \prod_{j = 1}^{d} p_i.c_j^{\alpha_j},$$
where $\alpha_j \ge 0$ is a function parameter and  $\sum_{j=1}^{d}\alpha_j \le 1$.
\end{defn}

\begin{defn}[$K$-Regret Query with MUFs]
The  \emph{$k$-regret query with MUFs} takes a database $\mathcal{D}$ of $d$-dimensional points and a family of MUFs $\mathcal{F}$ 
as the input. It returns a size-$k$ subset $\mathcal{S} \subseteq \mathcal{D}$, such that  
the maximum regret ratio $mr\_ratio_\mathcal{D}(\mathcal{S},\mathcal{F})$ is minimized.
\end{defn}

We note that a $k$-regret query to find the size-$k$ subset $\mathcal{S}$ that  minimizes
the maximum regret ratio is NP-hard, as shown by Chester et al.~\cite{chester2014computing}. 
In this study, we first focus on bounding the maximum regret ratio for $k$-regret queries with MUFs. 
We compute a subset $\mathcal{S}$ with a maximum regret ratio that is bounded by a 
decreasing function of $k$. This subset, however, may not minimize the maximum regret ratio. 
We thus further design a greedy algorithm to compute a subset to (heuristically) minimize the maximum regret ratio.

\textbf{Scale invariance.}
It has been shown~\cite{kessler2015k,nanongkai2010regret} that $k$-regret queries with additive utility functions are \emph{scale invariant}, i.e., 
scaling the data domain in any dimension does not change the maximum regret ratio of a set $\mathcal{S}$.
This property also holds for $k$-regret queries with MUFs.
For an MUF $f(p_i) = \prod_{d = 1}^{n} p_i.c_j^{\alpha_j}$, we can scale each dimension by
a factor $\lambda_j  > 0$, resulting in a new MUF $f'(p_i) = \prod_{d = 1}^{n} (\lambda_j \cdot p_i.c_j)^{\alpha_j} = 
\prod_{d = 1}^{n} \lambda_j^{\alpha_j} \cdot \prod_{d = 1}^{n}p_i.c_j^{\alpha_j}  = (\prod_{d = 1}^{n} \lambda_j^{\alpha_j}) f(p_i)$.
Such scaling does not affect the regret ratio (and hence the maximum regret ratio), i.e., 
 $r\_ratio_{\mathcal{D}}(\mathcal{S},f') = r\_ratio_{\mathcal{D}}(\mathcal{S},f)$:
\begin{align*}
r\_ratio_\mathcal{D}(\mathcal{S},f')&=\frac{max_{p_i\in \mathcal{D}}f'(p_i)-max_{p_j\in \mathcal{S}}f'(p_j)}{max_{p_i \in \mathcal{D}}f'(p_i)}\\
&=\frac{(\prod_{d = 1}^{n} \lambda_j^{\alpha_j})(max_{p_i\in \mathcal{D}}f(p_i)-max_{p_j\in \mathcal{S}}f(p_j))}{(\prod_{d = 1}^{n} \lambda_j^{\alpha_j})max_{p_i \in \mathcal{D}}f(p_i)}\\
&=\frac{max_{p_i\in \mathcal{D}}f(p_i)-max_{p_j\in \mathcal{S}}f(p_j)}{max_{p_i \in \mathcal{D}}f(p_i)}\\
&=r\_ratio_{\mathcal{D}}(\mathcal{S},f)
\end{align*}

In what follows, for conciseness, we refer to the regret of a query user as the regret when the context is clear. 
The same applies to the regret ratio and the maximum regret ratio of a query user.

\section{The MinVar Algorithm}\label{sec:minvar}
We propose an algorithm named \emph{MinVar} to process $k$-regret queries with MUFs. 
MinVar shares a similar overall algorithmic approach with that of CUBE~\cite{nanongkai2010regret} and MinWidth~\cite{kessler2015k} 
which were proposed to process $k$-regret queries with additive utility functions. The core idea of the algorithm is to partition the data 
space into multiple buckets that together enclose all data points, and return one data point in each bucket to form the answer set $\mathcal{S}$. 
For any utility function, the corresponding optimal data point $p^*$ must be in  some bucket. 
There is a point $s^*$ in this bucket that has been returned in $\mathcal{S}$. 
The distance between $p^*$ and $s^*$ is bounded by the range of attribute values spanned by the bucket in each attribute, 
which further bounds the maximum regret ratio. 

Our contributions in the MinVar algorithm are a novel data space partitioning strategy that follows the data distribution and 
produces tighter buckets, and theoretical bounds on the maximum regret ratios obtained. 
We first discuss how to partition the data space and select the point 
in each bucket to be returned. We derive the bounds on the maximum regret ratio afterwards.

\IncMargin{1em}
\begin{algorithm} 
\caption{MinVar} \label{alg:minvar}
\KwIn{$\mathcal{D}=\{p_1, p_2, ... , p_n\}$: a $d$-dimensional database;  $k$: the size of the answer set.} 
\KwOut{$\mathcal{S}$: a size-$k$ subset of $\mathcal{D}$.} 
$\mathcal{S} \leftarrow \emptyset$\;
\For {$i = 1, 2,..., d-1$} {
	Find $p^*_{i}$ which has the largest utility $ p^*_{i}.c_i$ in dimension $i$\;
	$c_i^\tau \leftarrow p^*_{i}.c_i$\; 
	$\mathcal{S} \leftarrow \mathcal{S} \cup \{p^*_{i}\}$\;
}
$itr \leftarrow 0$\;
\While{$|\mathcal{S}|< k$ and  $itr < itr_{m}$}{
 $t \leftarrow \lfloor (k-d+1)^{\frac{1}{d-1}} \rfloor + itr$\;
\For{$i=1,2,...,d-1$}{
	$bps[i] \leftarrow FindBreakpoints(\mathcal{D}, t, n, i, c_i^\tau)$\;
}
\For{each $(d-1)$-integer combination $1 \le j_1 \le t,  1 \le j_2 \le t, ..., 1 \le j_{d-1} \le t$}
{
$B  \leftarrow \{p \in \mathcal{D} | \forall i \in [1..d-1]: bps[i][j_i].lo \le p.c_i \le bps[i][j_i].hi \}$\;
$s^* \leftarrow argmax_{p\in B} \ p.c_d$\;
$\mathcal{S} \leftarrow \mathcal{S} \cup \{s^*\}$\;
	\If{$|\mathcal{S}|== k$}{
		break\;
	}
}
$itr \leftarrow irt+1$\;
}
$\mathcal{S} \leftarrow \mathcal{S} \cup \{k - |\mathcal{S}|$ random points in $\mathcal{D} \setminus \mathcal{S}\}$\;
\textbf{return} $\mathcal{S}$\;
\end{algorithm}
\DecMargin{1em}

\textbf{The MinVar algorithm}. Algorithm~\ref{alg:minvar} summarizes the MinVar algorithm. 
Before partitioning the data space,  
the algorithm first finds the optimal point $p^*_i$ in each of the 
first $d-1$ dimensions, i.e., $p_i^*.c_i$ is the largest  in dimension $i$ ($i = 1, 2, ..., d-1$). 
These $d-1$ optimal points are added to $\mathcal{S}$ (Lines~1~to~5). They are intended to minimize 
the regret ratio for the first $d-1$ dimensions.  
Another $k-d+1$ points are needed to fill up $\mathcal{S}$. These points minimize the regret ratio for dimension $d$.  
The algorithm partitions each of the first $d-1$ dimensions of the data domain into 
$t$ intervals (Lines 8 to 10), where 
\begin{equation}
t = \lfloor (k-d+1)^{\frac{1}{d-1}} \rfloor
\end{equation}
These intervals together partition the data space into $t^{d-1}$ buckets:
\begin{equation}
1 \le t^{d-1} = \lfloor (k-d+1)^{\frac{1}{d-1}} \rfloor^{(d-1)} \le k-d+1
\end{equation}
The algorithm selects one point $s^*$ 
in each bucket that has the largest utility $s^*.c_d$ in dimension $d$, and adds $s^*$ to $\mathcal{S}$ (Lines~11~to~14). 
There may be less than $k-d+1$ buckets, and some of the buckets may be empty. Thus, there may be less than $k-d+1$ points added in this step. 
To ensure $k$ points in $\mathcal{S}$, we repeat the partitioning step and increase $t$ by 1 in each iteration (Lines 7, 8,  and 17). This creates 
more buckets and obtains more points to be added to $\mathcal{S}$. The loop terminates when $|\mathcal{S}| = k$ (Lines 15 and 16) or a preset number 
of iterations $itr_{m}$ has been reached (Line 7). At this point, if $|\mathcal{S}| < k$, we simply fill it up with randomly selected points and return the set (Line 18).
The set $\mathcal{S}$ is then returned (Line 19).
 
\begin{figure}[h]
\centering
\includegraphics[width=2.5in]{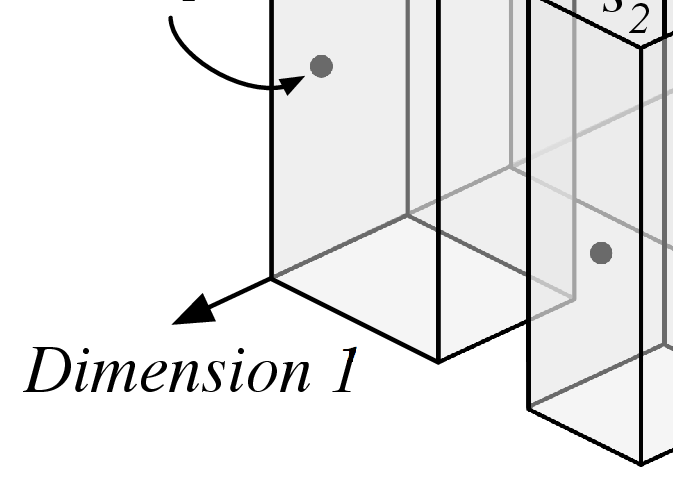}
\caption{The MinVar algorithm}\label{fig:minvar}
\end{figure}

Figure~\ref{fig:minvar} gives an example. Suppose $d = 3$ and $k = 6$. Then $t = \lfloor (6-3+1)^{\frac{1}{3-1}} \rfloor = 2$. 
We first add the two points $p_1^*$ and $p_2^*$ to $\mathcal{S}$ which have the largest utility in dimensions 1 and 2, respectively.
Then, the data domain in dimensions 1 and 2 are each partitioned into $t=2$ intervals, forming  
$t^{d-1}= 2^{3-1} = 4$ buckets. Four more points $s_1^*$, $s_2^*$, $s_3^*$, and $s_4^* $ are 
 added to $\mathcal{S}$, each has the largest utility in dimension 3 in a different bucket. 
Now there are 6 points in $\mathcal{S}$, i.e., $\mathcal{S} = \{p_1^*, p_2^*, s_1^*, s_2^*, s_3^*, s_4^* \}$. No further partitioning 
is needed, and the set $\mathcal{S}$ is returned as the query answer. 

\textbf{The FindBreakpoints algorithm.} 
The novelty of MinVar lies in the sub-algorithm \emph{FindBreakpoints} to find the breakpoints to partition a dimension 
of the data domain into $t$ intervals (Line 10). 
The intuition behind the algorithm is as follows.  
The optimal point $p^*$ for any utility function $f$ must lie in one of the buckets created.
Let this bucket be $\mathcal{B}$. The algorithm selects a point 
$s^*$ from $\mathcal{B}$ to represent this bucket and adds it to $\mathcal{S}$.
In the best case,  $p^*$ is selected as $s^*$, and the regret ratio is $0$. 
To maximize the probability of $p^*$ being selected, intuitively, 
we should partition each dimension such that every interval contains the same number 
of points. Otherwise, if the intervals are skewed and $p^*$ happens to lie in a dense interval, 
then its probability of being selected is small. 

\begin{figure}[h]
\centering
\includegraphics[width=2.5in]{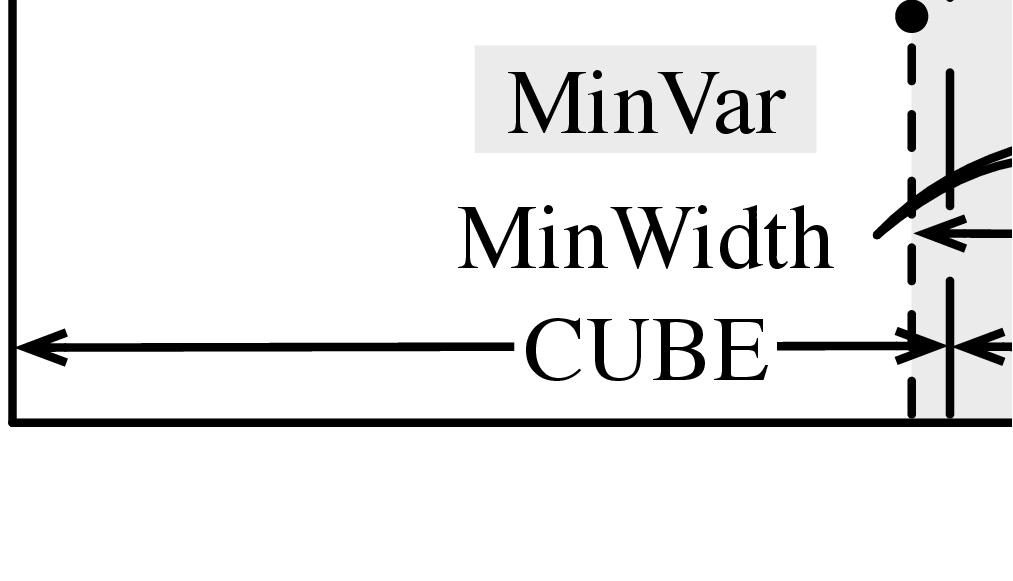}
\caption{Creating the intervals in a dimension}\label{fig:breakpoints}
\end{figure}
 
In CUBE~\cite{nanongkai2010regret}, the data domain is partitioned evenly, i.e., each interval has the same size of $\displaystyle \frac{c_i^\tau}{t}$, where $c_i^\tau$ denotes the largest utility in dimension $i$ (assuming that the data domain starts at 0).
When the data points are not uniformly distributed, 
the probability of $p^*$ being selected to represent its bucket is small. 
Figure~\ref{fig:breakpoints} gives an example where $d=2$. The two intervals created by CUBE in dimension 1 are highly unbalanced in the 
number of points in each interval. Point $p^*$ has a large utility in both dimensions and may be optimal for many MUFs. 
However, it will not be selected by CUBE, since it falls in a dense bucket, and there is 
another point $s^*$ with a larger utility in dimension 2.  
In MinWidth~\cite{kessler2015k}, the data domain is partitioned with a greedy heuristic. 
This heuristic leaves out some \emph{empty intervals} with no data points. 
It uses a binary search to determine the minimal interval width such that $t$ intervals of such width can cover the rest of the data domain.
MinWidth handles sparse data better, 
but the equi-width intervals still do not handle skewed data well. Figure~\ref{fig:breakpoints} shows two intervals created by MinWidth
which are still unbalanced (one has 5 points and the other has 3). The two points $p^*$ and $s^*$ are still in the same bucket.

\IncMargin{1em}
\begin{algorithm}
\caption{FindBreakpoints}\label{alg:bps}
\KwIn{$\mathcal{D}=\{p_1, p_2, ... , p_n\}$: a $d$-dimensional database; 
$t$: number of intervals; $n:$ size of $\mathcal{D}$; $i$: dimension number for which the breakpoints are to be computed;
$c_i^\tau$: the largest utility in dimension $i$.}
\KwOut{$bps[i]$: an array of $t$ pairs of breakpoints.}
Sort $\mathcal{D}$ ascendingly on dimension $i$; let the sorted point sequence be 
$p_1', p_2', ..., p_n'$\;
$hi \leftarrow 0,\ \delta \leftarrow 0$\;
\While{$hi \neq n$}{
$lo \leftarrow 1$\;
\For{j $=1,2, ..., t$}
{
$hi' \leftarrow \min\{lo+\lceil n/t\rceil-1+\delta, n\}$\;
Find the largest ${hi} \in [lo..hi']$ such that $p'_{hi}.c_i - p'_{lo}.c_i \displaystyle{\leq\frac{c_i^\tau}{t}}$\;
$bps[i][j].lo  \leftarrow p'_{lo}.c_i$\;
$bps[i][j].hi \leftarrow p'_{hi}.c_i$\;
$lo \leftarrow hi +1$\;
}
$\delta \leftarrow \delta + inc$\;
}
\textbf{return} $bps[i]$\;
\end{algorithm}
\DecMargin{1em} 

To overcome these limitations, our FindBreakpoints algorithm
adaptively uses $t$ variable-width intervals such that the number of points in each interval 
is as close to $\displaystyle \lceil \frac{n}{t}\rceil$ as possible, i.e., the \emph{var}iation of the number of points 
in each interval is \emph{min}imized, which motivates the name of the MinVar algorithm.  
As shown in Fig.~\ref{fig:breakpoints}, the two gray intervals created by 
MinVar contain $\displaystyle \lceil \frac{8}{2}\rceil = 4$ points each; $p^*$ will be selected to represent its bucket. 
To help derive the maximum regret ratio bounds in the following subsections, we also require that the width of each interval 
does not exceed $\displaystyle \frac{c_i^\tau}{t}$.
Under this constraint, it is not always possible to create $t$ intervals with exactly $\displaystyle \lceil \frac{n}{t}\rceil$ points in each interval.
Therefore, we allow $\displaystyle \lceil \frac{n}{t}\rceil + \delta$ data points in each interval, where $\delta$ is a parameter that will be 
adaptively chosen by the algorithm. At the start $\delta = 0$.

Algorithm~\ref{alg:bps} summarizes the FindBreakpoints algorithm. 
This algorithm first sorts the data points in ascending order of their coordinate values in dimension~$i$.
The sorted points are denoted as $p'_1, p'_2, ..., p'_n$ (Line~1).
The algorithm then creates $t$ intervals, where $lo$ and $hi$ represent the subscript lower and upper bounds of the data points to be put into one interval, respectively.
At the start, $lo = 1$ (Line~4). 
Between  $lo$ and $\displaystyle lo+ \lceil \frac{n}{t}\rceil-1 + \delta$, the algorithm finds the largest subscript $hi$ such that $p'_{hi}.c_i - p'_{lo}.c_i$ does 
not exceed  $\displaystyle{\frac{c_i^{\tau}}{t}}$. We then have obtained the two breakpoints 
of the first interval $bps[i][1].lo = p'_{lo}.c_i$ and $bps[i][1].hi = p'_{hi}.c_i$, where $bps[i]$ is an array to store the intervals in dimension $i$.
We update $lo$ to be $hi+1$, and repeat the above process to create 
the next interval (Lines 5 to 10). When $t$ intervals are created, if they cover all the $n$ points, we have successfully created the intervals for dimension~$i$. 
Otherwise, we need to allow a larger number of points in each interval. We increase $\delta$ by $inc$ which is 
a system parameter (Line 11), and repeat the above procedure to create $t$ intervals until $n$ points are covered. 
Then, we return the interval array $bps[i]$ (Line 12).
Note that the algorithm always terminates, because when $\delta$ increases to $n$, the algorithm will 
simply create intervals each with width $\displaystyle{\frac{c_i^\tau}{t}}$. The $t$ intervals must cover the entire data domain and hence cover 
all $n$ points. 

\textbf{Algorithm correctness.} 
As will be shown in the following subsections, the bounds on maximum regret ratios rely on the fact that the interval size does not exceed 
$\displaystyle \frac{c_i^\tau}{\lfloor (k-d+1)^{\frac{1}{d-1}} \rfloor}$.
In MinVar, even though FindBreakPoints creates variable-width intervals, each interval is still bounded by $\displaystyle \frac{c_i^\tau}{t}$. 
The value of $t$ starts at  $\lfloor (k-d+1)^{\frac{1}{d-1}} \rfloor$ and is kept increasing in the loop. 
Therefore, MinVar creates intervals where the size does not exceed $\displaystyle \frac{c_i^\tau}{\lfloor (k-d+1)^{\frac{1}{d-1}} \rfloor}$. This 
satisfies the requirement of the bounds and guarantees the algorithm correctness. 
In practice, MinVar can obtain maximum regret ratios smaller than the upper bound derived, since the intervals created by FindBreakPoints 
may be smaller than $\displaystyle \frac{c_i^\tau}{t}$ and $t$ may be larger than $\lfloor (k-d+1)^{\frac{1}{d-1}} \rfloor$.

\textbf{Algorithm complexity.} 
FindBreakpoints uses a database $\mathcal{D}$ of $n$ points and 
an array $bps[i]$ to store $t$ intervals. 
The space complexity is $O(n+t) = O(n)$ where $t=\lfloor (k-d+1)^{\frac{1}{d-1}} \rfloor \ll  n$. 
Leaving out the space for storing the input data, the space complexity is $O(t)$.
 Sorting the points in dimension $i$ takes $O(n\log n)$ time (Line 1).
The inner loop of FindBreakpoints (Lines 5 to 10) has $t$ iterations. In each iteration, computing $hi$ requires a binary search between $p'_{lo}$ and $p'_{hi'}$,
which takes $O(\log n)$ time. Thus, The inner loop takes $O(t\log n)$  time. The outer loop has $\displaystyle \frac{n}{inc}$ 
iterations in the worst case. Together, FindBreakpoints takes $O(\displaystyle n\log n + \frac{tn\log n}{inc})$ time.

MinVar uses a database $\mathcal{D}$ of $n$ points, an answer set $\mathcal{S}$ of size $k$, 
 a $(d-1)\times t$ two dimensional array $bps$. 
 An array of size $t^{d-1} = k-d+1$ is also needed to help  select the points $s^*$ in the $k-d+1$ buckets. 
 The space complexity is $O(n+k+dt)$. 
Leaving out the space for storing the input data, the space complexity is $O(k+dt)$.
The first for-loop of MinVar (Lines 2 to 5) takes $O(nd)$ time. 
The second for-loop (Lines 9 and 10) calls FindBreakpoints $d-1$ times, which takes $O(\displaystyle dn\log n +  \frac{tdn\log n}{inc})$ time.
The third for-loop (Lines 11 to 16) finds a point $s^*$ in each of the $k-d+1$ buckets. 
A linear scan on the database $\mathcal{D}$ is needed for this task. For each point $p$ visited, we need 
a binary search on each of the $d-1$ arrays $bps[i]$ to identify the bucket of $p$, and to update the point selected $s^*$ in that bucket 
if needed. This takes $O(nd\log t)$ time.
The second and third for-loops are enclosed in a loop to iterate through multiple values of $t$. 
The number of iterations is bounded by $itr_{m}$. 
Overall, MinVar takes $O(nd +  itr_{m}(\displaystyle dn\log n + \frac{(t+itr_{m})dn\log n}{inc} + nd\log (t+itr_{m}))) $ time. 
Here, $\displaystyle \frac{n}{inc}$ and $itr_m$ are controllable parameters of the system. 
In the experiments, we set $inc = 0.1\%n$. We observe that $itr_{max} = 11$ is sufficient in the data sets tested.
The time complexity then simplifies to $O(nd\log n + nd\log t) = O(nd\log (nt))$.

\section{Maximum Regret Ratio Bounds}\label{sec:bounds}
We derive bounds for the maximum regret ratio of $k$-regret queries with MUFs. 

\subsection{Upper Bound}\label{sec:upperbound}

We start with an upper bound of a set $\mathcal{S}$ returned by MinVar.
The intuition behind the bound is as follows. Given any MUF, its optimal point $p^*$ must be in some bucket 
created by MinVar. There is also one point $s^*$ returned by MinVar that is at the same bucket as $p^*$. The point $s^*$ has 
the largest coordinate value in dimension $d$. Meanwhile, the difference between $s^*$ and $p^*$ in dimension $i$ ($1 \le i \le d-1$) 
is bounded by the interval size $\displaystyle \frac{c_i^\tau}{t}$. Thus, the difference between the weighted products of the coordinate values 
of the two points should be bounded in a certain range. This range yields an upper bound of the maximum regret ratios. 

We assume that $p_i.c_j$ has been
normalized into the range of $(1,2]$ to simplify the derivation of the  upper bound.
This can be done by a normalization function $\displaystyle \mathcal{N}(p_i.c_j) = 1+ \frac{p_i.c_j}{\max_{p_i \in \mathcal{D}}\{p_i.c_j\}}$.
In fact, it is common to normalize the data domain in different dimensions into the same range, so that utility values of different dimensions
become more comparable. \emph{Note that our derivation of the bounds still holds without this assumption, although the bounds may become less concise.} 
This assumption does not affect the correctness of the MinVar algorithm either, although now the interval size should be bounded by 
$\displaystyle \frac{c_i^\tau - 1}{t}$ where $1$ is the lower bound of the data domain. 

Our upper bound is given by the following theorem.
\begin{thm} 
Let $\mathcal{F} = \{f| f(p_i)=\prod^d_{j=1} p_i.c_j^{\alpha_j}\}$ be a set of MUFs,  
where $\alpha_j\geq 0, \sum_{j=1}^{d}\alpha_j \le 1$, and $1 < p_i.c_j \le 2$.
The maximum regret ratio $mr\_ratio_{\mathcal{D}}(\mathcal{S}, \mathcal{F})$ of an answer set $\mathcal{S}$ 
 of MinVar satisfies 
\begin{equation}
mr\_ratio_{\mathcal{D}}(\mathcal{S}, \mathcal{F}) \le {\ln(1+\frac{1}{t})}
\end{equation}
\end{thm}

\begin{proof}
We prove the theorem by showing that for each MUF $f \in \mathcal{F}$, 
the regret ratio $r\_ratio_{\mathcal{D}}(\mathcal{S}, f)$ must be less than or equal to $\displaystyle \ln(1+\frac{1}{t})$. Thus, the maximum regret ratio of $\mathcal{F}$
must also be less than or equal to $\displaystyle \ln(1+\frac{1}{t})$.

Let $p^*$ be the point in $\mathcal{D}$ with the largest utility computed by $f$, i.e., 
$$p^*= \underset{p_i\in \mathcal{D}}{\mathrm{argmax}} f(p_i)$$ 
Let $s^*$ be the point in $\mathcal{S}$ that is selected by MinVar in the same bucket in which $p^*$ lies.
We have:
\begin{equation}\label{eq:ub1}
\begin{array}{l}
 regret_\mathcal{D}(\mathcal{S},f) = \max_{p_i\in\mathcal{D}}f(p_i) - \max_{p_i\in\mathcal{S}}f(p_i)\\
\displaystyle \quad \quad \quad \quad \quad \quad \le f(p^*)-f(s^*)\\
 \displaystyle  \quad \quad \quad \quad \quad \quad =\prod^d_{j=1} p^*.c_j^{\alpha_j}-\prod^d_{j=1} s^*.c_j^{\alpha_j}\\
 \displaystyle  \quad \quad \quad \quad \quad \quad =\exp{(\ln \prod^d_{j=1} p^*.c_j^{\alpha_j})}-\exp{(\ln \prod^d_{j=1} s^*.c_j^{\alpha_j})}
\end{array}
\end{equation}

Next, we show that $e^x-e^y\leq(x-y) e^x$, which will enable us to simplify the exponential terms in the equation above.
Let $g(x) = 1 - e^{-x} - x$. We have $g'(x) = e^{-x} -1$ and $g''(x) = -e^{-x}$. By letting $g'(x) = 0$, we have $x = 0$, while $g''(0) = -1 < 0$. Thus, the maximum of $g(x)$
is $g(0) = 0$, which means $1 - e^{-x} - x \le 0$. Therefore,  $1 - e^{-x} \le x$. Replacing $x$ by $x-y$ yields $1 - e^{y-x} \le x - y$. We multiply $e^x$ to both sides of the inequality 
and obtain $(1 - e^{y-x})e^x \le (x - y)e^x$. Thus, $e^x-e^y\leq(x-y) e^x$. Equation~\ref{eq:ub1} is then relaxed as follows.

\begin{equation}\label{eq:ub2}
\begin{array}{l}
\displaystyle regret_\mathcal{D}(\mathcal{S},f)  \le  (\ln \prod^d_{j=1} p^*.c_j^{\alpha_j} - \ln \prod^d_{j=1}  s^*.c_j^{\alpha_j})\cdot \prod_{j=1}^d  p^*.c_j^{\alpha_j}\\
\displaystyle  \quad \quad \quad \quad \quad \quad   =(\sum_{j=1}^d \alpha_j\ln  p^*.c_j-\sum_{j=1}^d \alpha_j\ln s^*.c_j)\cdot \prod_{j=1}^d p^*.c_j^{\alpha_j}\\
\displaystyle  \quad \quad \quad \quad \quad \quad  =\left[\sum_{j=1}^{d} \alpha_j(\ln  p^*.c_j-\ln s^*.c_j) \right]\cdot \prod_{j=1}^d p^*.c_j^{\alpha_j}
\end{array}
\end{equation}
Since MinVar selects the point in a bucket with the largest value in dimension $d$, 
we know that $ p^*.c_d\leq s^*.c_d$ and hence $\ln p^*.c_d\leq \ln s^*.c_d$, i.e., $\ln p^*.c_d -  \ln s^*.c_d \le 0$. Thus, we can remove 
the utility in dimension $d$ from the computation and relax the regret to be:  
\begin{equation}
\begin{array}{l}
\displaystyle regret_\mathcal{D}(\mathcal{S},f)
 \leq \left[\sum_{j=1}^{d-1} \alpha_j(\ln  p^*.c_j-\ln s^*.c_j) \right] \cdot \prod_{j=1}^d p^*.c_j^{\alpha_j} \\
\displaystyle  \quad \quad \quad \quad \quad \quad =\left[\sum_{j=1}^{d-1} \alpha_j \ln (\frac{ p^*.c_j}{s^*.c_j}) \right]\cdot \prod_{j=1}^d p^*.c_j^{\alpha_j}
\end{array}
\end{equation}
Since $s^*$ is selected from the same bucket in which $p^*$ lies, 
$p^*.c_j-s^*.c_j$ must be constrained by the bucket size in dimension $j$, which is $\displaystyle \frac{c_j^{\tau}-1}{t}$ where $c_j^{\tau}$ and 1
are the largest and smallest utility values in dimension $j$, i.e., 
\begin{equation}\label{eq:11}
\forall j \in [1..d-1], p^*.c_j-s^*.c_j\leq \frac{c_j^{\tau}-1}{t}
\end{equation}
Thus, 
\begin{equation}
\frac{p^*.c_j}{s^*.c_j}\leq1+\frac{c_j^\tau-1}{t\cdot s^*.c_j}
\end{equation}
Since $1 < s^*.c_j\leq c_j^\tau \le 2$, we have
\begin{equation}
\frac{p^*.c_j}{s^*.c_j}<1+\frac{1}{t}
\end{equation}
Therefore, 
\begin{equation}
regret_\mathcal{D}(\mathcal{S},f)< \left[\sum_{j=1}^{d-1} \alpha_j \ln (1+\frac{1}{t}) \right]\cdot \prod_{j=1}^d p^*.c_j^{\alpha_j}
\end{equation}
For the regret ratio $r\_ratio_\mathcal{D}(\mathcal{S},f)$, we have
\begin{equation}
\begin{array}{l}
\displaystyle r\_ratio_\mathcal{D}(\mathcal{S},f)
 =\frac{regret_\mathcal{D}(\mathcal{S},f)}{gain(\mathcal{D}, f)}\\
\displaystyle \quad \quad \quad \quad \quad \quad < \frac{ \left[\sum_{j=1}^{d-1} \alpha_j \ln (1+\frac{1}{t}) \right]\cdot \prod_{j=1}^d p^*.c_j^{\alpha_j}}{ \prod_{j=1}^d p^*.c_j^{\alpha_j}}\\
 \displaystyle  \quad \quad \quad \quad \quad \quad =\sum_{j=1}^{d-1} \alpha_j \ln (1+\frac{1}{t})=\ln\left((1+\frac{1}{t})^{\sum_{j=1}^{d-1}\alpha_j}\right)\\
\displaystyle  \quad \quad \quad \quad \quad \quad \leq \ln(1+\frac{1}{t})
\end{array}
\end{equation}
\end{proof}

In the theorem, $\displaystyle{t=\lfloor (k-d+1)^{\frac{1}{d-1}}\rfloor}$. 
Intuitively, when $k$ increases (i.e., returning more points), the maximum regret ratio is expected to decrease; when $d$ increases (i.e., accumulating regret over more dimensions), 
the maximum regret ratio is expected to increase. 
For simplicity, we say that the upper bound grows in the scale of $O(\ln(1+\frac{1}{k^{\frac{1}{d-1}}}))$.

To give an example, 
consider a 2-dimensional database, i.e., $d = 2$.
Let $k = 3$, which means $t = 2$. The upper bound of the maximum regret ratio is $\ln(1+\frac{1}{t}) = \ln \frac{3}{2} \approx 40.54\%$.
As $k$ increases (e.g., to $20$), this upper bound  will decrease (e.g., to $\ln \frac{20}{19} \approx 5.13\%$).

We have assumed  $p_i.c_j \in (1,2]$ in the proof above. In a more general case where $p_i.c_j$ lies in a range $(c_\bot, c_\top)$, $0 \le c_\bot < c_\top$ (positive utilities are considered), 
the upper bound derived becomes less concise. In particular, Equation~11 becomes $\displaystyle \forall j \in [1..d-1], p^*.c_j-s^*.c_j\leq \frac{c_j^{\tau}-c_\bot}{t}$, since 
the lower bound of the data space is now $c_\bot$. Dividing both sides of the inequality by $s^*.c_j$ yields:
 $$\displaystyle \frac{p^*.c_j}{s^*.c_j}\leq1+\frac{c_j^\tau-c_\bot}{t\cdot s^*.c_j} < 1+\frac{c_\top-c_\bot}{t\cdot c_\bot} $$
 The rest of the proof stays the same, except for that $\displaystyle 1+\frac{1}{t}$ needs to be replaced by $\displaystyle 1+\frac{c_\top-c_\bot}{t\cdot c_\bot}$.
 Therefore, in the more general case, the maximum regret ratio is bounded by $\displaystyle \ln(1+\frac{c_\top-c_\bot}{t\cdot c_\bot})$.
 This bound is less tight as it may be greater than 1. 

\subsection{Lower Bound}\label{sec:lowerbound}

We now derive a lower bound of the maximum regret ratio of any $k$-regret algorithm for MUFs (including but not limited to MinVar),  
assuming an infinite database $\mathcal{D}$.
We show that, given a family of MUFs $\mathcal{F}$, it is impossible 
to bound  the regret ratio to below $\displaystyle \Omega(\frac{1}{k^2})$ for a database $\mathcal{D}$ of infinite 2-dimensional points (i.e., $d=2$). 
The idea behind the lower bound is as follows. Given sufficient points, for any size-$k$ subset $\mathcal{S}$ returned, we can find an MUF $f$ such that its corresponding 
optimal point is sufficiently far away from any of the points in $\mathcal{S}$, and the regret ratio 
$r\_ratio_{\mathcal{D}}(\mathcal{S}, f)$ is at least $\displaystyle \Omega(\frac{1}{k^2})$. 

\begin{thm}
Given $k>0$, there must be a database $\mathcal{D}$ of 2-dimensional points such that 
the maximum regret ratio of any size-$k$ subset $\mathcal{S} \subseteq \mathcal{D}$
over a family of MUFs $\mathcal{F}$ is at least $\displaystyle \Omega(\frac{1}{k^2})$.
\end{thm}

\begin{figure}
\centering
\includegraphics[width=2.0in]{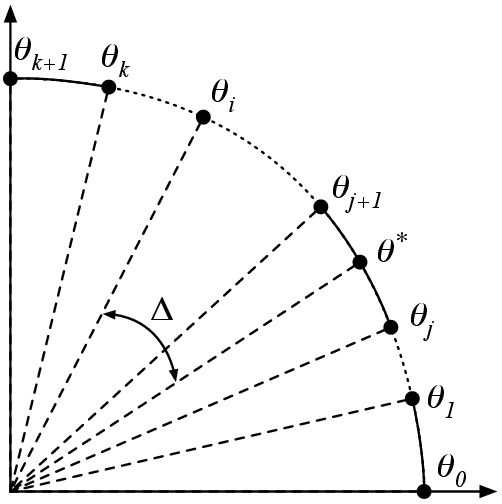}
\caption{Lower bound illustration}\label{fig:lower}
\end{figure}

\begin{proof}
We assume a data space of $(1,e] \times (1,e]$ in this proof.
Consider an infinite set $\mathcal{D}$ of 2-dimensional points, where each point $p$ corresponds to 
an angle $\theta$ in a polar coordinate system  as illustrated in Fig.~\ref{fig:lower}. 
The coordinate values of  $p$ satisfy: 
\begin{equation}
\begin{cases}
p.c_1=e^{\cos\theta}\\
p.c_2=e^{\sin\theta}
\end{cases}0 < \theta\leq\frac{\pi}{2}
\end{equation}
Given a size-$k$ subset $\mathcal{S} \subseteq \mathcal{D}$,  each point $s_i \in \mathcal{S}$ corresponds to a $\theta_i \in (0, \displaystyle \frac{\pi}{2}] $ 
 where $s_i.c_1=e^{\cos\theta_i}$ and $s_i.c_2=e^{\sin\theta_i}$.
Assume that the points $s_1, s_2, ..., s_k$ in $\mathcal{S}$ are sorted in ascending order of their corresponding $\theta_i$ values, i.e.,
$0< \theta_1 \le \theta_2 \le ... \le \theta_k \le \displaystyle \frac{\pi}{2}$.  
Further, let  $\theta_0=0$ and $\theta_{k+1}=\displaystyle \frac{\pi}{2}$. 
Then $\theta_i \ (i\in[0..k+1])$ can be represented as a point on a unit circle 
as shown in Fig.~\ref{fig:lower}. 
There are $k+2$ points in $\displaystyle [0, \frac{\pi}{2}]$.  
Every pair of adjacent points forms an angle, resulting in a total of $k+1$ angles. 
One of the adjacent pairs (i.e., for some $j\in[0..k]$) must satisfy: 
\begin{equation}
\theta_{j+1}-\theta_{j}\geq\frac{\pi}{2(k+1)}  
\end{equation}
Otherwise, if all angles are less than $\displaystyle \frac{\pi}{2(k+1)}$, their sum will be less than $\displaystyle \frac{\pi}{2}$.
Let $\theta^*$ be in the middle of $\theta_{j}$ and $\theta_{j+1}$, i.e., 
\begin{equation}
\theta^* = \displaystyle \frac{\theta_{j}+\theta_{j+1}}{2}
\end{equation}
We construct an MUF $f$ where the optimal point $p^*$ corresponds to $\theta^*$, i.e., 
$p^*.c_1=e^{\cos\theta^*}$ and $p^*.c_2=e^{\sin\theta^*}$, 
and prove the theorem based on the regret ratio of $f$.

Consider an MUF $f(p)=p.c_1^{\cos\theta^*}\cdot p.c_2^{\sin \theta^*}$.
\begin{equation}
\begin{array}{l}
\displaystyle \ln f(p)  =\ln{(p.c_1^{\cos\theta^*}\cdot p.c_2^{\sin\theta^*})}\\
\displaystyle   \quad \quad \quad  =\cos\theta^*\cdot \ln p.c_1+\sin\theta^*\cdot \ln p.c_2\\
 \displaystyle \quad \quad \quad =\cos\theta^*\cdot \cos\theta+\sin\theta^*\cdot \sin \theta
 \end{array}
\end{equation}
Let $g(\theta) = \cos\theta^*\cdot \cos\theta+\sin\theta^*\cdot \sin \theta$. 
By letting $g'(\theta) = 0$, we obtain 
$\theta = \theta^*$. Meanwhile, $g''(\theta^*) = - 1 < 0$. Thus, $\ln f(p)$ is maximized when $\theta = \theta^*$, and $f(p)$ is maximized when 
$p = p^*.$
\begin{equation}
\ln f(p^*) = \cos^2\theta^*+\sin^2\theta^* = 1; \quad \quad f(p^*) = e.
\end{equation}
Meanwhile, let $s_i$ be the optimal point for $f$ in $\mathcal{S}$. 
Since there is no other points in $\mathcal{S}$ that is between $\theta_j$ and $\theta_{j+1}$,
\begin{equation}
|\theta_i-\theta^*|=\Delta \geq \theta_{j+1} -\theta^* = \theta^* -\theta_j  \ge \frac{\pi}{4(k+1)}
\end{equation}
We consider the case where $\theta_i-\theta^*=\Delta$. 
The other case where $\theta^*-\theta_i=\Delta$ is symmetric. We omit it for conciseness.  
\begin{equation}
\begin{array}{l}
\displaystyle \ln f(s_i) =\ln{(s_i.c_1^{\cos\theta^*}\cdot s_i.c_2^{\sin\theta^*})}\\
\displaystyle \quad \quad \quad =\cos(\theta^*+\Delta)\cdot \cos\theta^*+\sin(\theta^*+\Delta)\cdot \sin\theta^*\\
\displaystyle \quad \quad \quad =(\cos\theta^*\cos\Delta-\sin\theta^*\sin\Delta)\cdot\cos\theta^*+
 (\sin\theta^*\cos\Delta+\cos\theta^*\sin\Delta)\cdot\sin\theta^*\\
\displaystyle \quad \quad \quad =\cos^2\theta^*\cdot\cos\Delta+\sin^2\theta^*\cdot\cos\Delta\\
\displaystyle \quad \quad \quad =\cos\Delta
\end{array}
\end{equation}
Here, the transformation is based on sine and cosine of sum identities. 
Now we have $f(s_i) = e^{\cos\Delta}$, and $r\_ratio_{\mathcal{D}}(\mathcal{S}, f)$ satisfies 
\begin{equation}
r\_ratio_{\mathcal{D}}(\mathcal{S}, f) = \frac{f(p^*) - f(s_i)}{f(p^*) } =  \frac{e-e^{\cos\Delta}}{e}=1-e^{\cos\Delta-1}
\end{equation}
Based on the Maclaurin series, we have
\begin{equation}
e^{\cos{\Delta}-1}=1-\frac{\Delta^2}{2}+\frac{\Delta^4}{6}-\cdots
\end{equation}
Thus,
\begin{equation}
r\_ratio_{\mathcal{D}}(\mathcal{S}, f) =\frac{\Delta^2}{2}-\frac{\Delta^4}{6}+\cdots
\end{equation}
We already know that
\begin{equation}
\Delta\geq\frac{\pi}{4(k+1)}
\end{equation}
Therefore, 
\begin{equation}
r\_ratio_{\mathcal{D}}(\mathcal{S}, f)  \ge \frac{\pi^2}{32(k+1)^2}-o(\frac{1}{k^4})
\end{equation}
This means that $r\_ratio_{\mathcal{D}}(\mathcal{S}, f)$ is at least $\displaystyle \Omega(\frac{1}{k^2})$.
\end{proof}

\section{Case Studies}\label{sec:casestudy}
 
 In this section, we showcase the applicability of the MinVar algorithm to both MUFs and non-MUFs. 
 We derive the maximum regret ratio bounds for applying MinVar on 
  $k$-regret queries with a real world example of MUFs -- the Cobb-Douglas family of utility functions,  
  and a closely related family of non-multiplicative utility functions -- 
  the Constant Elasticity of Substitution (CES) family of utility functions.

 \subsection{The K-Regret Query with Cobb-Douglas Functions}
 
The Cobb-Douglas function was first proposed as a production function to model the relationship between 
multiple inputs and the amount of output generated~\cite{10.2307/1811556}.
It was later generalized as a utility function. As a real example of  MUFs, this utility function has been 
used extensively in economics studies for modeling the diminishing marginal rate of substitution~\cite{10.2307/1811556,DIAMOND19801,vilcu2011geometric}.

\begin{defn}[Cobb-Douglas function]
A generalized \emph{Cobb-Douglas function}~\cite{vilcu2011geometric} with $d$ inputs $x_1, x_2, ..., x_d$ is a mapping $\mathcal{X}:\mathbb{R}_+^d\rightarrow\mathbb{R}_+$,
$$\mathcal{X}(x_1,x_2, ..., x_d)=A\prod^d_{j=1}x_j^{\alpha_j}$$
Here, $A>0$ and $\alpha_j\geq0$ are the function parameters. 
\end{defn}
The generalized Cobb-Douglas function is very similar to an MUF.
The $d$ inputs here can be seen as a data point of $d$ dimensions where input $x_j$ is the utility in dimension $j$. 
 MinVar applies to $k$-regret queries with Cobb-Douglas functions straightforwardly. 

To derive an upper bound 
of the maximum regret ratio for a family of  Cobb-Douglas functions 
$$\mathcal{F} = \{\mathcal{X}_1, \mathcal{X}_2, ..., \mathcal{X}_n\},$$
we transform each function $\mathcal{X}_i$ to an MUF by scaling the parameter $A$ to 1. It can be shown straightforwardly that this scaling does 
not affect the regret ratio or the maximum regret ratio.
Assume that $x_j$ has been normalized into the range of $(1,2]$. Then, the regret ratio upper bound derived in 
Section~\ref{sec:upperbound} applies, i.e., 
\begin{equation}
r\_ratio_\mathcal{D}(\mathcal{S}, \mathcal{X}_i) \le \ln\left((1+\frac{1}{t})^{\sum_{j=1}^{d-1}\alpha_j}\right)
\end{equation}
Here, each function $\mathcal{X}_i$ has a different set of parameters $\{\alpha_1,  \alpha_2, ..., \alpha_d\}$.
If $\sum^d_{j=1}\alpha_j\leq1$ holds for every $\mathcal{X}_i \in \mathcal{F}$, the maximum regret ratio is bounded by
\begin{equation}
mr\_ratio_\mathcal{D}(\mathcal{S}, \mathcal{F}) \le \ln(1+\frac{1}{t})
\end{equation}
Otherwise, the maximum regret ratio is bounded by
\begin{equation}
\begin{array}{l}
\displaystyle mr\_ratio_\mathcal{D}(\mathcal{S}, \mathcal{F}) \le \ln\left((1+\frac{1}{t})^{\alpha^\tau}\right), \\
\displaystyle \alpha^\tau =  \max \{\sum_{j=1}^{d-1}\mathcal{X}_i.\alpha_j|\mathcal{X}_i \in \mathcal{F}, \mathcal{X}_i.\alpha_j \text{ is a parameter of } \mathcal{X}_i\}.
\end{array}
\end{equation}
Similarly, the lower bound $\displaystyle \Omega(\frac{1}{k^2})$ of the maximum regret ratio derived in Section~\ref{sec:lowerbound} also applies.

\subsection{The K-Regret Query with CES Functions}
The CES function is a non-MUF that is closely related to the Cobb-Douglas function.
It is also used as a production function as well as a utility function~\cite{varian1992microeconomic,vilcu2011some}. 
The function provides an alternative model for how well the utility of an attribute makes up for that of another attribute, 
which is often used in economics studies~\cite{uzawa1962production,varian1992microeconomic,vilcu2011some} and 
has been considered previously for $k$-regret queries~\cite{kessler2015k}.

 \begin{defn}[CES function]
A generalized \emph{CES function}~\cite{vilcu2011some} with $d$ inputs $x_1, x_2, ..., x_d$ is a mapping $\mathcal{X}:\mathbb{R}_+^d\rightarrow\mathbb{R}_+$,
$$\mathcal{X}(x_1,x_2, ..., x_d)=A(\sum^d_{j=1}\alpha_j x_j^{\rho})^{\frac{\gamma}{\rho}}$$
Here, $A>0$, $\alpha_j\geq0$, $\rho<1$ ($\rho\ne 0$), and $\gamma>0$  are the function parameters. 
\end{defn}
When $\rho$ approaches 0 in the limit, the CES function will become a Cobb-Douglas function. When $\rho$ approaches 1, the CES function 
is very similar to the linear summation utility function. The case where $\rho > 1$ is not considered in the original proposal~\cite{10.2307/2232795} of the CES utility function.
We do not consider this case either, but this case could be an interesting subject for future work.

Algorithm MinVar  can also process $k$-regret queries with CES utility functions. To derive bounds 
for the maximum regret ratio, we simplify and rewrite the CES function $\mathcal{X}$ as a function $f$ as follows, assuming that 
$A = \gamma=1$. Making $A=1$ can be done by scaling, while the case where $\gamma \neq 1$ is considered as future work. 
$$f(p_i)=(\sum^d_{j=1}\alpha_j \cdot p_i.c_j^b)^{\frac{1}{b}}$$
Here, $0<b<1$ and $\alpha_j \ge 0 $.

It has been shown~\cite{kessler2015k} that the maximum regret ratio for $k$-regret queries with CES utility functions
is bounded between $\displaystyle \Omega (\frac{1}{bk^2})$ and $\displaystyle O(\frac{1}{bk^\frac{b}{d-1}})$ when $0 < b < 1$ 
(between $\displaystyle \Omega (\frac{1}{bk^2})$ and $\displaystyle O(\frac{1}{k^\frac{1}{b(d-1)}})$ when $b > 1$).
The lower bound also applies to our MinVar algorithm.
In what follows, we derive a tighter upper bond for the case where $0 < b < 1$. Note that this bound does not require the data space to be in $(1, 2]$ 
in each dimension, and it also applies to the MinWidth algorithm~\cite{kessler2015k}.

We first derive a new upper bound for the regret ratio for a single CES utility function $f$. 
Again, the intuition 
is to use the bucket size to bound the difference between $f(p^*)$ and $f(s^*)$, where $p^*$ is the optimal point for $f$ 
and $s^*$ is the point in the same bucket as $p^*$ returned by MinVar.

\begin{thm} Let $f(p_i)=(\sum^d_{j=1}\alpha_j \cdot p_i.c_j^b)^{\frac{1}{b}}$ be a CES utility function,  
where $0<b<1$ and $\alpha_j \ge 0 $.
The regret  ratio $r\_regret_{\mathcal{D}}(\mathcal{S}, f)$ of a set $\mathcal{S}$ 
 returned by MinVar satisfies 
\begin{equation}
r\_ratio_{\mathcal{D}}(\mathcal{S}, f) \le \displaystyle{\frac{d^{\frac{1}{b}}} {t+d^{\frac{1}{b}}}}
\end{equation}
\end{thm}

\begin{proof}
Let $p^*$ be the point in $\mathcal{D}$ with the largest utility computed by $f$, and
$s^*$ be the point in $\mathcal{S}$ that is selected in the same bucket in which $p^*$ lies.
We have:
\begin{equation}
\begin{array}{l}
\displaystyle regret_\mathcal{D}(\mathcal{S},f)  = \max_{p_i\in\mathcal{D}}f(p_i) - \max_{p_i\in\mathcal{S}}f(p_i)\\
\displaystyle \quad \quad \quad \quad \quad \quad \le f(p^*)-f(s^*)\\
\displaystyle \quad \quad \quad \quad \quad \quad  =(\sum_{j=1}^{d}\alpha_j \cdot p^*.c_j^b)^{\frac{1}{b}}-(\sum_{j=1}^{d}\alpha_j \cdot s^*.c_j^b)^{\frac{1}{b}}
\end{array}
\end{equation}
Since $g(x) = x^\frac{1}{b}$ is convex when $0 < b < 1$, we have 
$g(x) - g(y) \le (x-y)g'(x)$. Thus, 
\begin{equation}
\begin{array}{l}
\displaystyle regret_\mathcal{D}(\mathcal{S},f)
  \leq(\sum_{j=1}^{d}\alpha_j \cdot p^*.c_j^b-\sum_{j=1}^{d}\alpha_j \cdot s^*.c_j^b) \frac{1}{b} (\sum_{j=1}^{d}\alpha_j \cdot p^*.c_j^b)^{\frac{1}{b}-1}\\
\displaystyle  \quad \quad \quad \quad \quad \quad  = \frac{1}{b}\left[\sum_{j=1}^{d}\alpha_j(p^*.c_j^b-s^*.c_j^b) \right](\sum_{j=1}^{d}\alpha_j \cdot p^*.c_j^b)^{\frac{1}{b}-1}
\end{array}
\end{equation}
Consider another function $g(x) = x^b$, which is concave when $0<b<1$, and $g'(x)$ is monotonically decreasing ($g''(x) < 0$).
According to Lagrange's Mean Value Theorem, there must exist some value $\xi$ between two values $x$ and $y$, such that 
$g(x) - g(y) = x^b-y^b= (x - y)\cdot g'(\xi) = (x - y)\cdot b\cdot \xi ^{b-1}$. Further, since $g'(x)$ is monotonically decreasing, $b\cdot \xi ^{b-1} \le b\cdot (\min\{x, y\})^{b-1} $.
Thus, we have 
\begin{equation}
\begin{array}{l}
\displaystyle p^*.c_j^b-s^*.c_j^b
 \leq |p^*.c_j-s^*.c_j|\cdot b\cdot (\min\{p^*.c_j,s^*.c_j\})^{b-1}\\
\displaystyle \quad \quad \quad \quad \quad   \le \frac{c_j^\tau}{t}\cdot b\cdot {c_j^\tau}^{b-1} = \frac{b}{t}{c_j^\tau}^b
\end{array}
\end{equation}
 Therefore, 
\begin{equation}
\begin{array}{l}
regret_\mathcal{D}(\mathcal{S},f)
\displaystyle \le \frac{1}{b}(\sum_{j=1}^{d}\alpha_j \frac{b}{t}{c_j^\tau}^b ) (\sum_{j=1}^{d}\alpha_j \cdot p^*.c_j^b)^{\frac{1}{b}-1} \\
\displaystyle \quad \quad \quad \quad \quad \quad  \leq \frac{d}{t} (\max_{j \in [1..d]}\alpha_j {c_j^\tau}^{b} ) \left[d \cdot \max_{j \in [1..d]}\alpha_j{c_j^\tau}^b\right]^{\frac{1}{b}-1} \\
\displaystyle \quad \quad \quad \quad \quad \quad  = \frac{d^{\frac{1}{b}} }{t} (\max_{j \in [1..d]} \alpha_j{c_j^\tau}^b )^{\frac{1}{b}} \\
\end{array}
\end{equation}
Recall that $p^*_j$ has the largest utility in dimension $j$, i.e., $p^*_j.c_j = c_j^\tau$. This means that 
$\alpha_j p^*_j.c_j^b =\alpha_j{c_j^\tau}^b$. Since $\forall l \in [1..d], \alpha_l \ge 0$ and $p^*_j.c_l > 0$, we have $\sum_{l=1}^{d}\alpha_l\cdot {p^*_j.c_l}^b \ge \alpha_j p^*_j.c_j^b = \alpha_j{c_j^\tau}^b$. Thus,
\begin{equation}
\begin{array}{l}
regret_\mathcal{D}(\mathcal{S},f)
\displaystyle \leq \frac{d^{\frac{1}{b}} }{t} (\max_{j \in [1..d]}\sum_{l=1}^{d}\alpha_l\cdot {p^*_j.c_l}^b)^{\frac{1}{b}}\\
\displaystyle \quad \quad \quad \quad \quad \quad  =\frac{d^{\frac{1}{b}} }{t} \max_{j \in [1..d]}(\sum_{l=1}^{d}\alpha_l \cdot {p^*_j.c_l}^b)^{\frac{1}{b}}
\end{array}
\end{equation}
Let $\sigma = \displaystyle  \frac{d^{\frac{1}{b}} }{t}$. Then, 
\begin{equation}
\begin{array}{l}
\displaystyle regret_\mathcal{D}(\mathcal{S},f) \le \sigma \max_{j \in [1..d]}(\sum_{l=1}^{d}\alpha_l \cdot {p^*_j.c_l}^b)^{\frac{1}{b}}\\
\displaystyle \frac{1}{\sigma} regret_\mathcal{D}(\mathcal{S},f) \le \max_{j \in [1..d]}(\sum_{l=1}^{d}\alpha_l \cdot {p^*_j.c_l}^b)^{\frac{1}{b}}
\end{array}
\end{equation}
By the design of the MinVar algorithm, $p_j^* \  (j \in [1..d-1])$ is in $\mathcal{S}$. Meanwhile, 
the point with the largest utility in dimension $d$ in each bucket is also in $\mathcal{S}$, which means that
$p_d^*$ is also in $\mathcal{S}$. Thus, 
\begin{align*}
\max_{p_i\in \mathcal{S}}f(p_i) &\geq \max_{j \in [1..d]}f(p_j^*)  \\
&=  \max_{j \in [1..d]} ( \sum^d_{l=1}{\alpha_l \cdot p^*_j.c_l}^b)^{\frac{1}{b}} \\
&\ge \frac{1}{\sigma}regret_\mathcal{D}(\mathcal{S},f)
\end{align*}
This means $\displaystyle \frac{\max_{p_i\in \mathcal{S}}f(p_i)}{regret_\mathcal{D}(\mathcal{S},f)} \ge \frac{1}{\sigma}$. 
The regret ratio $r\_regret_{\mathcal{D}}(\mathcal{S}, f)$ is hence bounded by
\begin{equation}
\begin{array}{l}
\displaystyle r\_regret_{\mathcal{D}}(\mathcal{S}, f)
  =\frac{regret_\mathcal{D}(\mathcal{S},f)}{regret_\mathcal{D}(\mathcal{S},f)+\max_{p_i\in \mathcal{S}}f(p_i)}\\
\displaystyle \quad \quad \quad \quad \quad \quad \quad  = \frac{1}{1+{\frac{\max_{p_i\in \mathcal{S}}f(p_i)}{regret_\mathcal{D}(\mathcal{S},f)}}}\\
\displaystyle \quad \quad \quad \quad \quad \quad \quad   \leq \frac{1}{1+\frac{1}{\sigma}}=\frac{\sigma}{1+\sigma} = \frac{d^{\frac{1}{b}}} {t+d^{\frac{1}{b}}}
\end{array}
\end{equation}
\end{proof}

Therefore, given a set $\mathcal{F}$ of CES functions, the maximum regret ratio $mr\_regret_{\mathcal{D}}(\mathcal{S}, \mathcal{F})$ satisfies:
\begin{equation}
mr\_regret_{\mathcal{D}}(\mathcal{S}, \mathcal{F}) \leq \frac{d^{\frac{1}{b}}} {t+d^{\frac{1}{b}}}
\end{equation}

We can see from this bound that, 
when $k$ decreases or $d$ increases, the maximum regret ratio is expected to increase. 
For simplicity, we say that this bound 
grows in a scale of $O(\frac{1}{k^{\frac{1}{d-1}}})$. This bound is tighter than the bound $O(\frac{1}{bk^\frac{b}{d-1}})$ obtained in 
a previous study~\cite{kessler2015k} since $0<b<1$.

To give an example, 
consider a family of CES functions where $b = 0.5$. Let $d = 2$ and $k = 3$, which means $t = 2$. The upper bound of the maximum regret ratio $\frac{d^{\frac{1}{b}}} {t+d^{\frac{1}{b}}} = 
\frac{4} {2+4} \approx 66.67\%$.
As $k$ increases (e.g., to $20$), this upper bound  will decrease (e.g., to $ \frac{4} {19+4} \approx 17.39\%$).

\section{The MaxDif Algorithm}\label{sec:maxdif}

MinVar aims to \emph{bound} the maximum regret ratios rather than to \emph{minimize} them. 
Its bucket-based answer point selection strategy is conservative. 
To minimize the maximum regret ratios, 
we propose a heuristic based second query algorithm named \emph{MaxDif} that exploits  \emph{skyline points}~\cite{borzsony2001skyline}.
We first show that the answer set to minimize the maximum regret ratio must be formed by  skyline points.
When there are more than $k$ skyline points, we need to select $k$ from them to form an answer set. 
MaxDif makes this selection following a heuristic for regret ratio minimization. 
If there are no more than $k$ skyline points, the entire set of skyline points should be returned as the answer. 
The answer set can be padded with randomly selected objects from the database to make it of size-$k$. 

\begin{figure}
\centering
\includegraphics[width=2.5in]{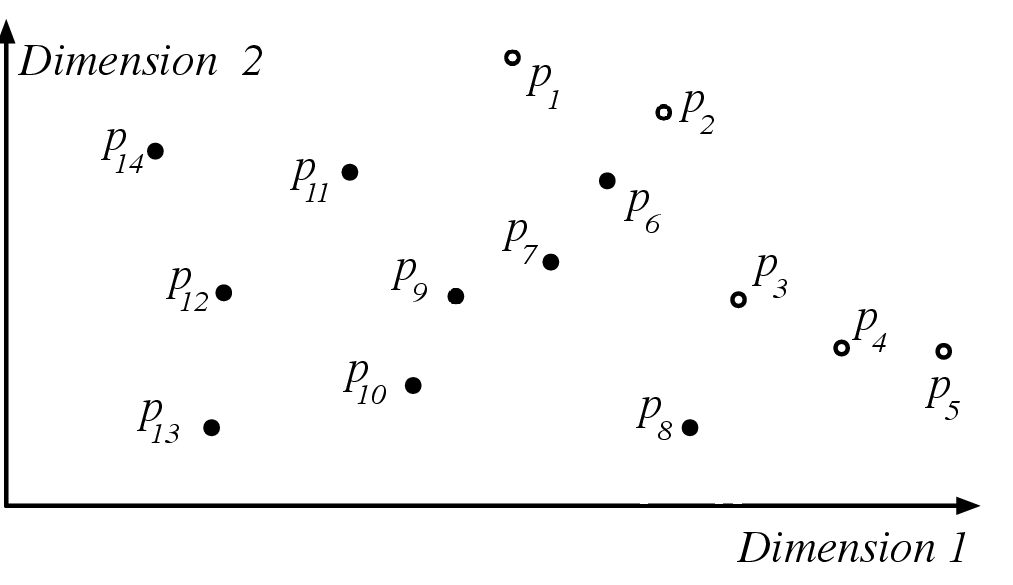}
\caption{Skyline point selection}\label{fig:selectpoints}
\end{figure}

\textbf{Regret ratio minimization with skyline points.}
Skyline points are points that are not \emph{dominated} by any other points.
 Given two points $p_i$ and $p_j$,  
  $p_i$ is said to dominate $p_j$ if and only if the coordinate values of $p_i$ are no smaller than those of $p_j$ in all dimensions, and there is at least one 
dimension where the coordinate  of $p_i$ is larger than that of $p_j$, i.e., 
\begin{equation}
\forall l \in [1..d], p_i.c_l \ge p_j.c_l \wedge \exists l \in [1..d], p_i.c_l > p_j.c_l
\end{equation}
In Fig.~\ref{fig:selectpoints}, the hollow dots $p_1$, $p_2$, $p_3$, $p_4$, and $p_5$ are skyline points as they are not dominated, while  
the black dots are \emph{non-skyline points}, i.e., points dominated by some skyline point. 

Let $\mathcal{P}$ be the set of all skyline points. Its regret  
over any MUF $f$, $regret_\mathcal{D}(\mathcal{P},f)$, must be $0$. This is because, for any non-skyline point $p_j \in \mathcal{D} \setminus \mathcal{P}$, the skyline point $p_i \in \mathcal{P}$ 
dominating $p_j$ satisfies $f(p_i) \ge f(p_j)$ for any MUF $f$ by definition, and hence $gain(\mathcal{P}, f) \ge gain(\mathcal{D} \setminus \mathcal{P}, f)$. Thus, 
\begin{equation}
\begin{array}{l}
\displaystyle regret_\mathcal{D}(\mathcal{P},f) 
 = gain(\mathcal{D},f) - gain(\mathcal{P},f)\\
\displaystyle \quad \quad \quad \quad \quad \quad = gain(\mathcal{D}\setminus \mathcal{P} \cup \mathcal{P},f) - gain(\mathcal{P},f)\\
 \displaystyle  \quad \quad \quad \quad \quad \quad= \max\{gain(\mathcal{D\setminus\mathcal{P}},f), gain(\mathcal{P},f)\} - gain(\mathcal{P},f)\\
 \displaystyle  \quad \quad \quad \quad \quad \quad=gain(\mathcal{P},f) - gain(\mathcal{P},f)\\
 \displaystyle  \quad \quad \quad \quad \quad \quad = 0
\end{array}
\end{equation}
Therefore, \emph{the regret ratio of $\mathcal{P}$ and the maximum regret ratio of $\mathcal{P}$  over a family of MUFs must be 0.}

When $|\mathcal{P}| \le k$, $\mathcal{P}$ is the optimal answer to a $k$-regret query with MUFs. 
When $|\mathcal{P}| > k$, a size-$k$ subset of it must be the optimal answer. The reason is as follows. 
For any size-$k$ subset $\mathcal{S} \in \mathcal{D}$, 
we can construct a new size-$k$ subset $\mathcal{S}'$ by replacing every non-skyline point $p_j\in \mathcal{S}$ with a skyline point $p_i \in \mathcal{P}$ 
that dominates $p_j$. The gain of  $\mathcal{S}'$ over any MUF $f$ must be larger than or equal to that of  $\mathcal{S}$. 
Thus, the maximum regret ratio of $\mathcal{S}'$ must be smaller than or equal to that of   $\mathcal{S}$.
Therefore, the optimal answer set that minimizes the maximum regret ratio must be a set that contains  skyline points only, and hence it is a subset of $\mathcal{P}$. 
This is formulated as the following theorem. 

\begin{thm}\label{thm:skyline} 
Let $\mathcal{F} = \{f| f(p_i)=\prod^d_{j=1} p_i.c_j^{\alpha_j}\}$ be a set of MUFs,  
where $\alpha_j\geq 0$ and  $\sum_{j=1}^{d}\alpha_j \le 1$. Suppose that there are $k$ or more skyline points in a database $\mathcal{D}$ of size $n$.
Then, there is a size-$k$ subset $\mathcal{S}^* \subseteq \mathcal{D}$ 
that contains only skyline points, and its maximum regret ratio is less than or equal to that of any other size-$k$ subset $\mathcal{S} \subseteq \mathcal{D}$, 
i.e., $mr\_ratio_{\mathcal{D}}(\mathcal{S}^*, \mathcal{F}) \le mr\_ratio_{\mathcal{D}}(\mathcal{S}, \mathcal{F})$.
\end{thm}

\begin{proof}
The proof is straightforward as sketched in the paragraph above. We omit the full proof for conciseness. 
\end{proof}

We note that a similar theorem has been proven~\cite{Asudeh:2017:ECR:3035918.3035932} in parallel to show that 
skyline points minimize the maximum regret ratio for AUFs. Both AUFs and MUFs are monotonically increasing 
in each dimension, and their parameters are both non-negative. This allows skyline points to 
minimize the maximum regret ratio for both types of utility functions. 

\textbf{Selection of $k$ skyline points.} 
Theorem~\ref{thm:skyline} reduces the $k$-regret query to selecting $k$ skyline points that together 
 minimize the maximum regret ratio. 
This can be done by finding the size-$k$ subset $\mathcal{S}^* \subseteq \mathcal{P}$ such that, 
 for any MUFs in $\mathcal{F}$,  the maximum regret of $\mathcal{S}^*$ over 
 the gain of $\mathcal{P}$ is minimized.
 Since $\mathcal{P} =( \mathcal{P} \setminus \mathcal{S}^*) \cup \mathcal{S}^*$, the goal translates to 
  minimizing the maximum ratio of gain difference between $\mathcal{P}\setminus \mathcal{S}^*$  and 
  $\mathcal{S}^*$  over the gain of  $\mathcal{P}\setminus \mathcal{S}^*$ for any MUFs in $\mathcal{F}$. Formally, 
\begin{equation}\label{eq:sstar}
\begin{array}{l}
\displaystyle \mathcal{S}^*  = \underset{\mathcal{S}\subseteq \mathcal{P}, |\mathcal{S}| = k}{\arg\min} \max_{f\in\mathcal{F}} \{\frac{gain(\mathcal{P} \setminus \mathcal{S}, f) - gain(\mathcal{S}, f)}{gain(\mathcal{P} \setminus \mathcal{S}, f)}\}\\
\displaystyle \quad  = \underset{\mathcal{S}\subseteq \mathcal{P}, |\mathcal{S}| = k}{\arg\min} \max_{f\in\mathcal{F}} \{\frac{\max_{p_i\in \mathcal{P} \setminus \mathcal{S}} f(p_i) - \max_{p_j\in \mathcal{S}} f(p_j)}{\max_{p_i\in \mathcal{P} \setminus \mathcal{S}} f(p_i)}\}
 \end{array}
\end{equation}
There are $\displaystyle {|\mathcal{P}| \choose k} = \frac{|\mathcal{P}|!}{k!(|\mathcal{P}|-k)!}$ size-$k$ subsets of skyline points.
Finding the optimal size-$k$ subset from them has been shown to be NP-hard~\cite{chester2014computing}.  
We use a greedy heuristic to select $k$ skyline points iteratively to form the answer set $\mathcal{S}^*$. 
In each iteration, the point $p^* \in \mathcal{P}$ is selected such that, for any MUF $f\in \mathcal{F}$, 
the maximum ratio of difference between $gain(\mathcal{P}\setminus\{p^*\}, f)$ and $f(p^*)$ over $gain(\mathcal{P}\setminus\{p^*\}, f)$ is minimized.
Point $p^*$ is added to $\mathcal{S}^*$ and removed from $\mathcal{P}$ before the next point is selected from $\mathcal{P}$. 
Formally,
\begin{equation}\label{eq:pstar}
\displaystyle p^* = \underset{p\in \mathcal{P}}{\arg\min} \max_{f\in\mathcal{F}} \{\frac{\max_{p_i\in \mathcal{P} \setminus \{p\}} f(p_i) -  f(p)}{\max_{p_i\in \mathcal{P} \setminus \{p\}} f(p_i)}\}
\end{equation}
As we can see, Equations~\ref{eq:sstar} and~\ref{eq:pstar} are very similar, except that a size-$k$ subset $\mathcal{S}$ has been replaced by a point $p$. 

To compute point $p^*$, we need to first rewrite Equation~\ref{eq:pstar}. We know that 
\begin{equation}
\max_{p_i\in \mathcal{P} \setminus \{p\}} f(p_i) -  f(p) = \max_{p_i\in \mathcal{P} \setminus \{p\}} \{f(p_i) -  f(p)\}
\end{equation}
Further, 
\begin{equation}
\displaystyle \frac{\max_{p_i\in \mathcal{P} \setminus \{p\}} \{f(p_i) -  f(p)\}}{\max_{p_i\in \mathcal{P} \setminus \{p\}} f(p_i) }
= \max_{p_i\in \mathcal{P} \setminus \{p\}}\{ \frac{f(p_i) -  f(p)}{f(p_i)} \}
\end{equation}
Thus, 
\begin{equation}
p^* = \underset{p\in \mathcal{P}}{\arg\min} \max_{f\in\mathcal{F}} \{\max_{p_i\in \mathcal{P} \setminus \{p\}} \{\frac{f(p_i) -  f(p)}{f(p_i)}\}\}
\end{equation}
The two ``$\max$'' aggregates in the equation above can be swapped without affecting $p^*$. They require checking all combinations of points in 
$\mathcal{P} \setminus \{p\}$ and functions in $\mathcal{F}$. The order of the checking has no impact on the computation result. Thus, 
\begin{equation}\label{eq:pstar1}
p^* = \underset{p\in \mathcal{P}}{\arg\min} \max_{p_i\in \mathcal{P} \setminus \{p\}} \{ \max_{f\in\mathcal{F}}\{\frac{f(p_i) -  f(p)}{f(p_i)}\}\}
\end{equation}
The two aggregates ``$\underset{p\in \mathcal{P}}{\arg\min}$'' and ``$\displaystyle \max_{p_i\in \mathcal{P} \setminus \{p\}}$'' in this equation 
can be handled simply by a two-layer loop to examine all the points in $\mathcal{P}$. 

\textbf{MaxDif computation.}  
The only problem remaining is to compute 
the term $\displaystyle \max_{f\in\mathcal{F}}\{\frac{f(p_i) -  f(p)}{f(p_i)}\}$.
This term represents the \emph{max}imum ratio of \emph{dif}ference between the utilities of $p_i$ and $p$ over the utility of $p_i$ 
for any MUFs in $\mathcal{F}$.
We call it the \emph{MaxDif} of $p_i$ over $p$, and denote it by $maxdif(p_i, p)$. 
\begin{equation}
maxdif(p_i, p) = \max_{f\in \mathcal{F}} \{\frac{f(p_i) -  f(p)}{f(p_i)}\}
\end{equation}
Without knowing the exact MUFs in  $\mathcal{F}$, however, it is 
infeasible to compute $maxdif(p_i, p)$.
We address this problem by computing an upper bound for it instead, which is denoted by $maxdif^*(p_i, p)$: 
\begin{equation}
maxdif^*(p_i, p) =  \max_{l \in [1..d]} \ln \frac{p_i.c_l}{p.c_l} 
\end{equation}
We show that $maxdif(p_i, p) \le  maxdif^*(p_i, p)$ holds by considering  
$\displaystyle \frac{f(p_i) -  f(p)}{f(p_i)}$ for any MUF $f\in \mathcal{F}$.
Equation~\ref{eq:ub2} 
in the proof of the maximum regret ratio upper bound in Section~\ref{sec:upperbound} suggests:
\begin{equation}
f(p_i)-f(p) \le \left[\sum_{l=1}^{d} \alpha_l(\ln  p_i.c_l-\ln p.c_l) \right] \prod_{l=1}^d p_i.c_l^{\alpha_l}
\end{equation}
By definition, $f(p_i) = \prod_{l=1}^d p_i.c_l^{\alpha_l} > 0 $. Thus,
\begin{equation}
\displaystyle \frac{f(p_i) -  f(p)}{f(p_i)} \le \sum_{l=1}^{d} \alpha_l(\ln  p_i.c_l-\ln p.c_l) 
\end{equation}
Since $\ln  p_i.c_l-\ln p.c_l \le \max_{l\in[1..d]} \{\ln  p_i.c_l-\ln p.c_l\}$, 
\begin{equation}
\begin{array}{l}
\displaystyle \frac{f(p_i) -  f(p)}{f(p_i)}
   \le \sum_{\gamma=1}^{d} \alpha_\gamma \max_{l \in [1..d]} \{\ln  p_i.c_l-\ln p.c_l\} \\
\displaystyle   \quad \quad \quad \quad \quad  = \sum_{\gamma=1}^{d} \alpha_\gamma \max_{l \in [1..d]} \ln  \frac{ p_i.c_l}{ p.c_l} 
\end{array}
\end{equation}
We know that $\sum_{\gamma=1}^{d} \alpha_\gamma \le 1$. Thus,
\begin{equation}
\displaystyle \frac{f(p_i) -  f(p)}{f(p_i)}  \le  \max_{l \in [1..d]} \ln  \frac{p_i.c_l}{ p.c_l}
\end{equation}
The right hand side of the inequality is independent of any MUF $f$. 
Thus, 
\begin{equation}
\begin{array}{l}
\displaystyle maxdif(p_i, p) 
= \max_{f\in \mathcal{F}} \{\frac{f(p_i) -  f(p)}{f(p_i)}\}\\
\displaystyle  \quad \quad \quad \quad \quad \quad \le \max_{l \in [1..d]} \ln  \frac{p_i.c_l}{ p.c_l} \\
\displaystyle  \quad \quad \quad \quad \quad \quad = maxdif^*(p_i, p)
\end{array}
\end{equation}
We can now replace $ \max_{f\in\mathcal{F}}\{\frac{f(p_i) -  f(p)}{f(p_i)}\}$ with $maxdif^*(p_i, p)$ in Equation~\ref{eq:pstar1} for computing $p^*$: 
\begin{equation}\label{eq:pstarub}
p^* = \underset{p\in \mathcal{P}}{\arg\min} \max_{p_i\in \mathcal{P} \setminus \{p\}} \{maxdif^*(p_i, p)\}
\end{equation}

\IncMargin{1em}
\begin{algorithm} 
\caption{MaxDif} \label{alg:maxdif}
\KwIn{$\mathcal{D}=\{p_1, p_2, ... , p_n\}$: a $d$-dimensional database;  $k$: the size of the answer set.} 
\KwOut{$\mathcal{S}^*$: a size-$k$ subset of $\mathcal{D}$.} 
Compute skyline points and store them in $\mathcal{P}$\;
$\mathcal{S}^* \leftarrow \emptyset$\;
\For {$i = 1, 2,..., d$} {
	Find $p^*_{i} \in \mathcal{P}$ which has the largest utility $ p^*_{i}.c_i$ in dimension $i$\;
	$\mathcal{S}^* \leftarrow \mathcal{S} \cup \{p^*_{i}\}$\;
	$\mathcal{P} \leftarrow \mathcal{P} \setminus \{p^*_i\}$\;	
}
\While{$|\mathcal{S}|< k$ and  $|\mathcal{P}| > 0$}{
	$p^* \leftarrow FindMinMaxDifPoint(\mathcal{P})$\;
	$\mathcal{S}^* \leftarrow \mathcal{S}^* \cup \{p^*\}$\;
	$\mathcal{P} \leftarrow \mathcal{P} \setminus \{p^*\}$\;	
}
$\mathcal{S}^* \leftarrow \mathcal{S}^* \cup \{k - |\mathcal{S}^*|$ random points in $\mathcal{D} \setminus \mathcal{S}^*\}$\;
\textbf{return} $\mathcal{S}^*$\;
\end{algorithm}
\DecMargin{1em}

\IncMargin{1em}
\begin{algorithm} 
\caption{FindMinMaxDifPoint} \label{alg:maxdifpoint}
\KwIn{$\mathcal{P}$: a skyline point set.} 
\KwOut{$p^*$: a MinMaxDif point.} 
$min \leftarrow +\infty$\;
\For {$p \in \mathcal{P}$} {
	$p.md \leftarrow -\infty$\;
	\For {$p_i \in \mathcal{P}\setminus \{p\}$} {
		\If{$maxdif^*(p_i , p$) $> p.md$ }{
			$p.md \leftarrow maxdif^*(p_i , p)$\;
		}
	}	
	\If{$p.md < min$ }{
		$p^* \leftarrow p$\;
		$min \leftarrow p.md$\;
	}
}
\textbf{return} $p^*$\;
\end{algorithm}
\DecMargin{1em}

\textbf{The algorithm.}
We name our query algorithm after the MaxDif metric, i.e., the \emph{MaxDif} algorithm.
As summarized in Algorithm~\ref{alg:maxdif},  the MaxDif algorithm first computes all the skyline points and 
stores them in a set $\mathcal{P}$  (Line~1). This can be done by an existing skyline query algorithm (e.g.,~\cite{papadias2003optimal,tan2001efficient}) 
and is not the focus of our study. A straightforward algorithm is a three-layer nested loop over all the data points 
and dimensions to look for any non-dominated points. 
Then, following MinVar, the algorithm adds the skyline point with the largest coordinate value 
in each dimension to the answer set $\mathcal{S}^*$ (Lines 3 to 6). 
This  serves to cover the extreme case where the MUFs have a weight of 1 in some dimension and 0's in all other dimensions. 
The algorithm proceeds to add point $p^*$ as defined by Equation~\ref{eq:pstarub} into $\mathcal{S^*}$ iteratively (Lines 7 to 10). 
We call point $p^*$ a \emph{MinMaxDif} point and use a sub-algorithm \emph{FindMinMaxDifPoint} to compute it. 
Each MinMaxDif point added to $\mathcal{S}^*$ is removed from $\mathcal{P}$, and the loop terminates when $\mathcal{S}^*$ has $k$ points or 
$\mathcal{P}$ becomes empty. If the loop terminates and $\mathcal{S}^*$ does not have $k$ points, we fill up $\mathcal{S}^*$ with randomly 
selected points from $\mathcal{D}$ (Line 11). Then, the set $\mathcal{S}^*$ is returned (Line 12).

The FindMinMaxPoint algorithm loops through the skyline points in $\mathcal{P}$. For every skyline point $p$, 
we compute the MaxDif value of every skyline point $p_i \in \mathcal{P}\setminus \{p\}$ over $p$. The largest 
MaxDif  value is recorded as $p.md$. The skyline point $p^*$ with the smallest $p^*.md$ value is returned 
as the MinMaxDif point. We summarize this process as Algorithm~\ref{alg:maxdifpoint}.

\textbf{Algorithm complexity.}
The MaxDif algorithm needs to compute and store the set of skyline points $\mathcal{P}$.
Leaving out the space for storing the input data, the space complexity of the algorithm is $O(|\mathcal{P}| +k)$.
In the worst case, $\mathcal{P}$ has the same size as the entire database, and the worst-case 
space complexity is $O(n +k)$.

Computing the set $\mathcal{P}$ with a straightforward three-layer nested loop takes $O(n^2d)$ time (Line~1). 
There are more advanced skyline query algorithms~\cite{papadias2003optimal,tan2001efficient} 
but these are beyond the scope of the paper. 
Computing the maximum skyline points in the $d$ dimensions (Lines 3 to 6) takes $O(|\mathcal{P}|d)$ time. 
The MaxDif algorithm then calls FindMinMaxDifPoint for $k$ times (Lines 7 to 10). 
FindMinMaxDifPoint makes a two-layer nested loop pass over $\mathcal{P}$ to compute the MinMaxDif point, where 
computing the MaxDif value between two points needs to loop through $d$ dimensions. Thus, the time complexity 
 for the $k$ function calls is $O(k|\mathcal{P}|^2d)$. The overall time complexity is 
$O(n^2d + |\mathcal{P}|d + k|\mathcal{P}|^2d) = O(n^2d+k|\mathcal{P}|^2d)$. The worst-case time complexity is $O(kn^2d)$.

\textbf{Discussion.}
The set $\mathcal{S}^*$ returned by the MaxDif algorithm is a heuristic choice to approach  
the theoretically optimal answer set defined by Equation~\ref{eq:sstar}. It is based on a bound $maxdif^*(p_i, p)$ and 
aims to minimize the maximum regret ratio over the entire set of utility functions in $\mathcal{F}$ which is infinite.
 In the experiments, we can only test a finite subset $\mathcal{F}' \subset \mathcal{F}$
of utility functions. 
The set $\mathcal{S}^*$ generated to minimize the maximum regret ratio over $\mathcal{F}$ may 
not minimize that over $\mathcal{F}'$. 
The reason is as follows. Given a finite set of utility functions $\mathcal{F}' \subset \mathcal{F}$, there may be a size-$k$ subset $\mathcal{S}$ that 
contains all the skyline points that maximize the gains over $\mathcal{F}'$, and hence minimizes the maximum regret ratio over $\mathcal{F}'$.
This subset $\mathcal{S}$, however, may not minimize the maximum regret ratio over the set in $\mathcal{F} \setminus \mathcal{F}' $. 
The MaxDif algorithm, which considers both sets of $\mathcal{F}'$ and $\mathcal{F} \setminus \mathcal{F}' $ together,  
may return a different set $\mathcal{S}^*$ . 
Since $\mathcal{S}^*$ is different from $\mathcal{S}$, and $\mathcal{S}$ minimizes the maximum regret ratio over $\mathcal{F}' $, 
$\mathcal{S}^*$ may not minimize the maximum regret ratio over $\mathcal{F}' $. 
Regardless, as the experiments in the next section show, $\mathcal{S}^*$ still has consistently small maximum regret ratios over the set $\mathcal{F}' $.
Further, a larger $\mathcal{F}' $ may cause only a small increase in the maximum regret ratio of  $\mathcal{S}^*$ as the MaxDif algorithm already considers the 
infinite set $\mathcal{F}$ when generating $\mathcal{S}^*$.

\section{Experiments}\label{sec:exp}

We evaluate the empirical performance of the two proposed algorithms MinVar  and MaxDif. 

\subsection{Settings}
The algorithms are implemented in C++, and the experiments are run on a computer running the OS X 10.12 operating system with a 64-bit 
2.7 GHz Intel\textsuperscript{\textregistered} Quad-Core\textsuperscript{(TM)} i7 CPU and 16 GB RAM. 

Both real and synthetic data sets are used in the experiments. 
The real data sets used are the \emph{NBA}\footnote{http://www.databasebasketball.com}, the \emph{Stocks}\footnote{http://pages.swcp.com/stocks}, 
and the \emph{Weather}\footnote{https://crudata.uea.ac.uk/cru/data/hrg/tmc/} data sets.
NBA and Stocks have been used in previous studies on $k$-regret queries~\cite{peng2014geometry,peng2015k}. 
After filtering out data points with null fields, we obtain 20,640 data points of 7 dimensions in the NBA data set, including 88 skyline points. 
The Stocks data set contains 122,574 data points of 5 dimensions, including 39 skyline points.  
Weather is a larger data set, which contains 566,262 data points of 13 dimensions, including 7,947 skyline points. The 13 dimensions of each data point represent   
the elevation and 12 monthly mean temperature values of a weather observation point. We use absolute values of the temperature data since we assume positive utilities. A value of  zero 
is also allowed as it does not affect the correctness of any algorithm tested.

The synthetic data sets are generated using the \emph{anti-correlated} data set generator~\cite{borzsony2001skyline}, which is a popular data generator 
used in skyline query studies~\cite{papadias2003optimal,pei2007probabilistic,tao2009distance}. 
This data generator can generate points with correlated, anti-correlated, and random coordinate values in different dimensions. Data points with 
correlated coordinate values have similar coordinate values in different dimensions. This means that a data point $p$ with the largest coordinate value in one dimension 
is likely to have large coordinate values in other dimensions as well, and $p$ tends to dominate most other points. Only a small number of points like $p$ are needed to dominate all other points in a data set. Such a data set has only a small number of skyline points. 
In contrast, data points with anti-correlated coordinate values have large coordinate values in some dimensions while small coordinate values in 
other dimensions.  They tend not to dominate or be dominated by other points, which makes them likely skyline points. More skyline points exist in such a data set. 
Data points with random coordinate values have independent
coordinate values in different dimensions. A data set of such points has a relatively moderate number of skyline points.   
We generate \emph{Correlated}, \emph{Anti-correlated}, and \emph{Random} data sets with these different type of points, respectively.

We vary the data set cardinality $n$ from 10,000 to 1,000,000, the dimensionality $d$ from 2 to 12, and the query parameter $k$ from 10 to 50. 
Table~\ref{tbl:setting} summarizes the parameters and their values.
By default, we use a Random data set with 100,000 data points of 5 dimensions ($d = 5$), and $k=20$.
Note that both a proposed algorithm MinVar and a baseline algorithm MinWidth~\cite{kessler2015k}   
divide the data space into $t^{(d-1)}$ buckets and select a single point from each bucket to be added into the query answer, where 
$\displaystyle t = \lfloor (k-d+1)^\frac{1}{d-1}\rfloor$. With a default $d$ value of 5, $t = 1$ (i.e., the entire data space is considered as a bucket) for $k$ up to $19$. 
For such small values of $k$, the performance difference between the two algorithms is minimum, which can be seen from the experimental results in Section~\ref{sec:results} where the value of $k$ is varied. 
To observe the performance difference of the two algorithms, we use a default value of $k = 20$. We argue that a representative subset of 20 data points is still manageable  
by users.

\begin{table}
\centering
\caption{Experimental Settings}\label{tbl:setting}
\BlankLine\BlankLine
\begin{tabular}{ccc}
\toprule
\multicolumn{1}{c}{Parameter} & Values & Default \\
\midrule
Utility function & AUF, CES, Cobb-Douglas  & - \\
Data set & NBA, Stocks, Weather, & Random\\
 & Anti-correlated, Correlated, Random  & \\
Number of utility functions & 10k, 50k, 100k, 500k, 1000k & 10k\\
$n$ & 10k, 50k, 100k, 500k, 1000k & 100k \\
$d$ & 2, 4, 5, 6, 8, 10, 12 & 5\\
$k$ & 10, 20, 30, 40, 50 & 20\\
\bottomrule
\end{tabular}
\end{table}

We use three families of utility functions -- the generalized Cobb-Douglas functions (denoted by ``\textbf{Cobb-Douglas}''), 
the CES functions (denoted by ``\textbf{CES}''), and the linear summation functions (denoted by ``\textbf{AUF}''). 
The involvement of AUFs in the experiments serves to showcase the applicability of the proposed MinVar and MaxDif algorithms over a wider range of utility functions. 
MinVar has the same maximum regret ratio bounds for AUFs as those derived by~Nanongkai et al.~\cite{nanongkai2010regret}, 
since MinVar also uses parameter $t$  to bound the space partitioning, the values of which are no smaller than those used 
by~Nanongkai et al. MaxDif can also handle $k$-regret queries with CES functions or AUFs, 
but may produce suboptimal query answers. This is because the optimization function of MaxDif is designed 
for MUFs which may not minimize the maximum regret ratios for CES functions or AUFs.
In each set of experiments, we randomly generate from 10,000 to 1,000,000 sets of parameters $\{\alpha_i| i\in[1..d], \alpha_i \in [0, 1]\}$  
for each family of utility functions, where $\sum_{i=1}^{d}\alpha_i = 1$. 
The CES function has an extra parameter $b$. We generate random values of $b$ in the range of $[0.1,0.9]$.
By default, we use 10,000 utility functions in each utility function family for the testing. 
We run the algorithms on the data sets, and report the running time and the maximum regret ratio (denoted by ``\textbf{MRR}'') on the generated utility functions.

Since no algorithms have been proposed in the past for $k$-regret queries with MUFs,  
for comparison purposes, we use four baseline algorithms MinWidth, Angle, AreaGreedy, and MaxDom. 
MinWidth, Angle, and AreaGreedy have been proposed for $k$-regret queries with CES functions~\cite{kessler2015k};
MaxDom has been proposed for top-$k$ representative skyline queries~\cite{lin2007selecting}. 
We compare the maximum regret ratios 
of the answer sets generated by these algorithms over Cobb-Douglas functions, CES functions, and AUFs  
with those of the answer sets generated by the two proposed algorithms. 
Together we test six algorithms in our experiments. 

\begin{itemize}
\item \emph{MinVar} is the algorithm proposed in Section~\ref{sec:minvar}. We use $inc = 0.1\%n$ 
to control the number of iterations for which the sub-algorithm FindBreakpoints is run. 
We find that $itr_{max}=11$ is sufficient to handle the data sets tested. The results obtained are based on this setting.

\item 
\emph{MaxDif} is the heuristic algorithm proposed in Section~\ref{sec:maxdif}.

\item \emph{MinWidth}~\cite{kessler2015k}  is an algorithm with bounded maximum regret ratios for $k$-regret queries 
with CES functions.

\item \emph{Angle}~\cite{kessler2015k}  is a greedy algorithm (with no bounds on the maximum regret ratios) proposed 
for $k$-regret queries with CES functions. 

\item \emph{AreaGreedy}~\cite{kessler2015k}  is another greedy algorithm (with no bounds on the maximum regret ratios) proposed 
for $k$-regret queries with CES functions. 

\item \emph{MaxDom}~\cite{lin2007selecting} is a greedy algorithm 
that returns the $k$ representative skyline points which dominate the largest number 
of other points. For fairness, both MaxDom and MaxDif use the same straightforward algorithm to compute the skyline points, which 
checks every pair of points and every dimension for dominance (with early termination once dominance is detected). Skyline computation time 
is included in the running time reported.
\end{itemize}

\subsection{Results}\label{sec:results}

\begin{figure}[htp]
\centering
	\subfloat[MRR (Cobb-Douglas)]{
	\hspace{5mm}
		\begin{overpic}[scale=1.06]{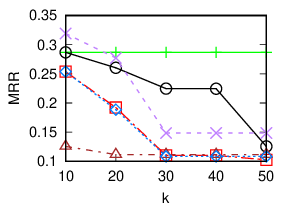}
		\end{overpic}
	}
	\subfloat[MRR (CES)]{
		\begin{overpic}[scale=1.06]{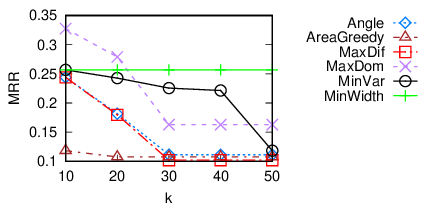}
		\end{overpic}
	}\\
		\vspace{-3mm}
	\subfloat[MRR (AUF)]{
		\hspace{-15mm}
		\begin{overpic}[scale=1.06]{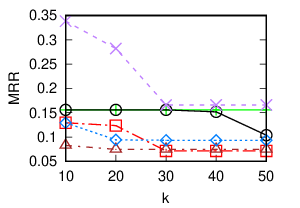}
		\end{overpic}
	}
	\subfloat[Running time]{
		\begin{overpic}[scale=1.06]{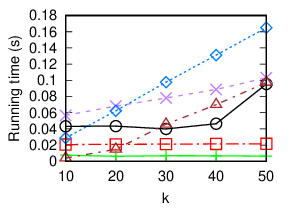}
		\end{overpic}
	}
	\vspace{-1mm}
	\caption{Varying $k$ (NBA) \label{fig:k_nba}}
\end{figure}

\textbf{Effect of $k$.}
We first test the effect of varying $k$. Figures~\ref{fig:k_nba} to~\ref{fig:k_random} show the result where $k$ 
is varied from 10 to 50 on the three real data sets and three synthetic data sets. 
In general, we can see decreasing maximum regret ratios (``MRR'' in the figures) as $k$ increases for all algorithms tested. Meanwhile, the algorithm 
running times increase. These are expected. A larger $k$ means more points are returned and a higher probability of satisfying more utility functions, 
which also take more time to compute. Note that the computation process and the answer set of each algorithm are independent of the different types of 
utility functions. Thus, the running times of the algorithms are independent of the utility function type, and we only report them once for each set of experiments.
The same answer set of each algorithm, however, may have different 
maximum regret ratios for different types of utility functions as shown in Fig.~\ref{fig:k_nba}, which is consistent with our motivating example in Section~\ref{sec:intro}.

\emph{Performance on real data sets.} 
Figures~\ref{fig:k_nba} to~\ref{fig:k_weather} show the results on the NBA, Stocks, and Weather data sets. 
We see that, between the two algorithms MinVar (denoted by`$\circ$' in the figure; same below) and MinWidth (`$+$') that have bounded maximum regret ratios, 
 the proposed algorithm MinVar has maximum regret ratios that are consistently lower than or equal to those of MinWidth.   
The advantage is more significant for larger $k$ values (e.g., up to 56\% lower for $k=50$ with Cobb-Douglas functions, cf. Fig.~\ref{fig:k_nba}a), 
as this allows more buckets to be created, and MinVar is designed 
to increase the probability of catching the optimal points by balancing the numbers of points across different buckets. 
This advantage comes at the expense of a higher running time than that of MinWidth.  
However, we argue that the running time of MinVar is manageable. 
As Fig.~\ref{fig:k_weather} shows,  MinVar can process the Weather data set which has over half a million data points in just 1.3 seconds. 
We also notice that MinWidth can process this data set in 0.3 seconds while producing the same maximum regret ratios. 
This shows that MinWidth is a highly competitive baseline algorithm, and it is not easy to outperform this algorithm in terms of the running time. 

\begin{figure}[htp]
\centering
	\subfloat[MRR (Cobb-Douglas)]{
		\hspace{5mm}
		\begin{overpic}[scale=1.06]{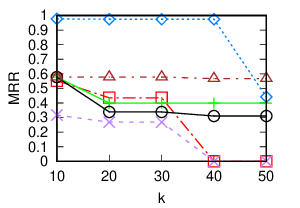}
		\end{overpic}
	}
	\subfloat[MRR (CES)]{
		\begin{overpic}[scale=1.06]{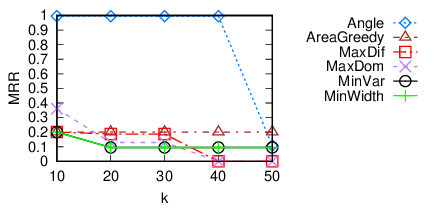}
		\end{overpic}
	}\\
	\vspace{-3mm}
	\subfloat[MRR (AUF)]{
		\hspace{-15mm}
		\begin{overpic}[scale=1.06]{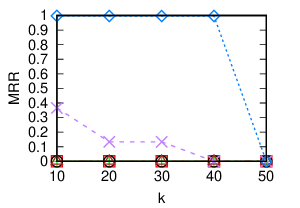}
		\end{overpic}
	}
	\subfloat[Running time]{
		\begin{overpic}[scale=1.08]{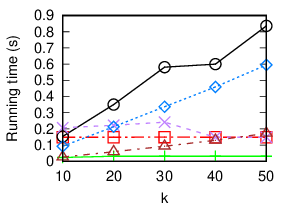}
		\end{overpic}
	}
	\vspace{-1mm}
	\caption{Varying $k$ (Stocks) \label{fig:k_stock}}
\end{figure}

\begin{figure}[htp]
\centering
	\subfloat[MRR (Cobb-Douglas)]{
		\hspace{5mm}
		\begin{overpic}[scale=1.06]{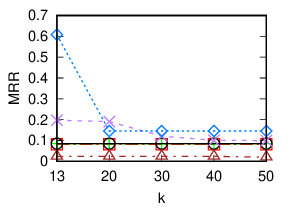}
		\end{overpic}
	}
	\subfloat[MRR (CES)]{
		\begin{overpic}[scale=1.06]{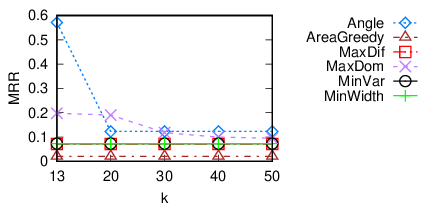}
		\end{overpic}
	}\\
	\vspace{-3mm}
	\subfloat[MRR (AUF)]{
		\hspace{-15mm}
		\begin{overpic}[scale=1.06]{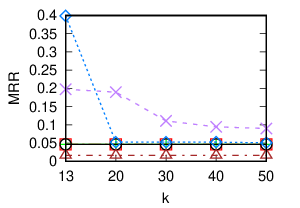}
		\end{overpic}
	}
	\subfloat[Running time]{
		\begin{overpic}[scale=1.07]{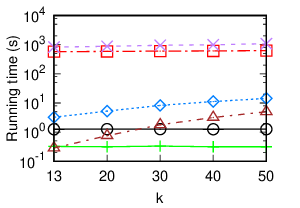}
		\end{overpic}
	}
	\vspace{-1mm}
	\caption{Varying $k$ (Weather) \label{fig:k_weather}}
\end{figure}

For the rest of the algorithms which do not have bounded maximum regret ratios, the proposed algorithm MaxDif (`$\Box$') has the most consistent performance in terms of 
the difference from the smallest maximum regret ratio achieved by any algorithm tested. 
The maximum regret ratios of MaxDif are no more than 0.24 higher than those of any algorithm tested (i.e., for Cobb-Douglas functions on the Stocks data set where $k=10$). 
MaxDif also has the smallest maximum regret ratios for the NBA data set where $k \ge 30$ and for the Stocks data set where $k \ge 40$. 
For the other heuristic algorithms, AreaGreedy (`$\triangle$') has maximum regret ratios that can be 
0.57 higher than that of MaxDif and MaxDom (`$\times$') on the Stocks data set where $k \ge 40$, although it also has 
the smallest maximum regret ratios on the Weather data set and for a few $k$ values on the NBA data set. 
Similarly, Angle (`$\diamond$') and MaxDom also suffer on the Stocks data set, where their maximum regret ratios are up to 0.99 and 0.36 higher than those 
of MaxDif, respectively.
The advantage of MaxDif is attributed to its point selection strategy. MaxDif adds the skyline point that differs the least from any of the unselected skyline points into the  answer set. By doing so, even if 
an unselected skyline point turns out to be optimal for some utility function, the optimality of the points in the answers set is not too much worse than that of the optimal point. 
This point selection strategy is also the reason why MaxDif may have larger maximum regret ratios than those of the other proposed algorithm 
MinVar in some cases (e.g., $k=$ 20 or 30 on the Stocks data set with Cobb-Douglas or CES  functions), 
where MinVar may happen to select an optimal skyline point or some non-skyline point close to it, while 
MaxDif may select a non-optimal skyline point. 

In terms of the running time, all the heuristic algorithms are slower than MinWidth. 
AreaGreedy and Angle require multiple scans over the data set to find the points that bound the maximum area and are the maximum 
towards different angles, respectively. MaxDif has a time complexity that is quadratic to the number of skyline points. Its running time is low when 
the number of skyline points is small, e.g., below 0.15 seconds for the NBA and Stocks data sets where the number of skyline points is below 100. 
This running time could become higher than those of AreaGreedy and Angle 
when the number of skyline points becomes larger, e.g., over 500 seconds for the Weather data set where  the number of skyline points is over 7,000.  
MaxDom checks every skyline point against all the data points and its running times are constantly higher than those of MaxDif.  It may take over 1,000 seconds to run on the Weather data set.   

\begin{figure}[htp]
\centering
	\subfloat[MRR (Cobb-Douglas) - NBA]{
		\hspace{5mm}
		\begin{overpic}[scale=1.06]{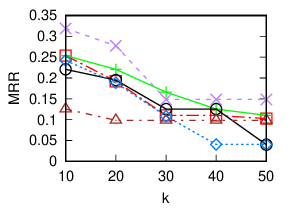}\label{fig:k_skyline_nba_mrr}
		\end{overpic}
	}
	\subfloat[Running time - NBA]{
		\begin{overpic}[scale=1.06]{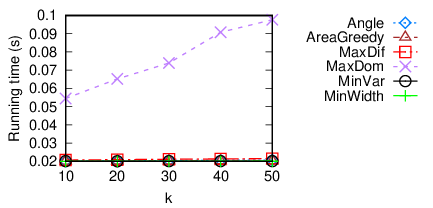}
		\end{overpic}
	}\\
		\vspace{-3mm}
	\hspace{-15mm}
	\subfloat[MRR (Cobb-Douglas) - Stocks]{
		\begin{overpic}[scale=1.06]{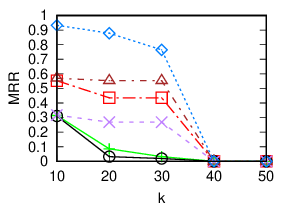}\label{fig:k_skyline_stock_mrr}
		\end{overpic}
	}
	\subfloat[Running time - Stocks]{
		\begin{overpic}[scale=1.07]{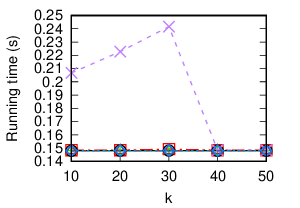}
		\end{overpic}
	}\\
		\vspace{-3mm}
		\hspace{-15mm}
	\subfloat[MRR (Cobb-Douglas) - Weather]{
		\begin{overpic}[scale=1.06]{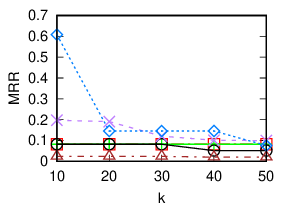}\label{fig:k_skyline_weather_mrr}
		\end{overpic}
	}
	\subfloat[Running time - Weather]{
		\begin{overpic}[scale=1.07]{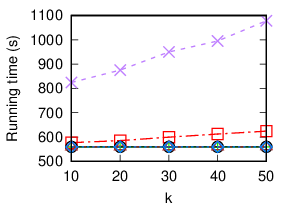}\label{fig:k_skyline_weather_time}
		\end{overpic}
	}
	\vspace{-1mm}
	\caption{Varying $k$ (skyline points of real data sets) \label{fig:k_skyline}}
\end{figure}

\emph{Performance on skyline points.}
The proposed algorithm MinVar and three baseline algorithms MinWidth, Angle, and AreaGreedy are designed to run on the entire data set, but they can also run 
on the set of of skyline points. We examine the performance of these four algorithms on skyline points in this set of experiments. 
Note that, for MinVar and MinWidth which have bounded maximum regret ratios, running them on skyline points 
does not impact their maximum regret ratio bounds for the following reason. We have shown that, for any MUF, its corresponding optimal point must be a skyline point. 
For any skyline point, a point is selected into the query answer from the bucket of the skyline point by MinVar and MinWidth. 
Further, the bucket size for the set of skyline points $\mathcal{P}$ must not exceed that  for the entire data set $\mathcal{D}$, 
since $\mathcal{P} \subseteq \mathcal{D}$. Thus, the bucket size based maximum regret ratio bounds still hold. We omit the full proof since it is straightforward. 

Figure~\ref{fig:k_skyline} shows the maximum regret ratios for Cobb-Douglas utility functions and running times when the algorithms are run on the skyline points of the three real datasets. 
We omit the MRR figures for CES and AUF utility functions, since the comparative performance of the algorithms running on skyline points with their counterparts running on the entire data set is similar to 
those shown in the MRR figures for Cobb-Douglas utility functions in Figs.~\ref{fig:k_nba} to~\ref{fig:k_skyline}. 

In terms of the maximum regret ratio, we see that MinVar (`$\circ$') and MinWidth (`$+$') benefit the most from running on the skyline points.  Their maximum regret ratios are now either the smallest (Fig.~\ref{fig:k_skyline_stock_mrr}) or very 
close to the smallest maximum regret ratios produced by any algorithm tested (Fig.~\ref{fig:k_skyline_nba_mrr} and Fig.~\ref{fig:k_skyline_weather_mrr}). Among these two algorithms, 
the proposed algorithm MinVar again obtains smaller maximum regret ratios because it can adaptively shrink the bucket size which leads to smaller  regret ratios. 
Angle (`$\diamond$') and AreaGreedy (`$\triangle$') benefit less. Their maximum regret ratios are still less stable across different data sets than those of the proposed algorithm MaxDif (`$\Box$'). 
They have the largest maximum regret ratios on the Stocks data set where $k \le 30$ (Fig.~\ref{fig:k_skyline_stock_mrr}), although AreaGreedy also has the smallest maximum regret ratios 
on the Weather data set (Fig.~\ref{fig:k_skyline_weather_mrr}). Focusing on the two proposed algorithms MaxDif and MinVar, 
MaxDif still produces smaller maximum regret ratios on the NBA data set where $20 \le k \le 40$ (Fig.~\ref{fig:k_skyline_nba_mrr}), while MinVar produces no larger 
maximum regret ratios in all other cases. This suggests that, while MaxDif is still an effective heuristic based algorithm, MinVar could be very competitive if it could be run on the skyline points. 

The smaller maximum regret ratios come at the cost of larger algorithm running times. Now all algorithms except MaxDom (`$\times$', which needs to run on the entire data set) 
have roughly the same running time, which is dominated by the time 
to compute the skyline points. For example, on the Weather data set (Fig.~\ref{fig:k_skyline_weather_time}), the running times of all algorithms except MaxDom are at about 560 seconds, among which 
 559.31 seconds are taken to compute the skyline points. Once the skyline points are computed, even  MaxDif  which 
has a quadratic running time on the number of skyline points only takes 64.88 seconds (i.e., a total of 624.19 seconds) to compute the query answer, whereas MinVar only takes 0.28 seconds 
to compute the query answer ($k=50$). 

In application scenarios where the data set is dynamic (e.g., online shopping services where new products keep arriving) 
and precomputing the skyline points is infeasible, the high cost of computing the skyline points 
would prevent running the algorithms on them. Thus, in the following experiments,  for MinVar, MinWidth, Angle, and AreaGreedy which 
do not have to run on the skyline points, we focus on their performance over the entire dataset.

\emph{Performance on synthetic data sets.} 
Similar performance patterns are observed on the three synthetic data sets, as shown in Figs.~\ref{fig:k_correlated} to~\ref{fig:k_random}.
MinVar (`$\circ$') has maximum regret ratios that are smaller than or equal to those of MinWidth~(`$+$') in almost all cases, except for 
when $k = 50$ on the Anti-correlated data set with CES functions (Fig.~\ref{fig:k_anticorrelated}b). The advantage in 
the maximum regret ratio is most significant on the Correlated data set. This can be explained by noting that the Correlated data set has a more skewed distribution, 
and MinVar is designed to obtain more balanced buckets for skewed data. The running times of MinVar are again larger than those of MinWidth, but are within 0.2 seconds, which is still 
reasonably small. Following studies~\cite{Miller:1968:RTM:1476589.1476628,Shneiderman:1984:RTD:2514.2517} 
on users' tolerable waiting times, we consider 2 seconds as the threshold of a ``reasonable'' running time.
MaxDif (`$\Box$')  has the smallest maximum regret ratios (up to 0.019) and running times (up to 0.030 seconds) 
on the Correlated data set except for when $k = 10$ (cf.~Fig.~\ref{fig:k_correlated}). 
This data set has 26 skyline points, which can be processed by MaxDif with a high efficiency. 
MaxDom (`$\times$') has the smallest maximum regret ratios on the Anti-correlated (up to 0.238) and Random data sets (up to 0.146), 
but its running times are also the largest (up to 569.3 and 9.8 seconds on the two data sets, respectively).  
MaxDif has the second smallest maximum regret ratios on these two data sets (up to 0.326 and 0.224 which are 37.0\% and 53.4\% higher, respectively) 
with the exception of AUFs on the Anti-correlated data set. Note that 
the optimization goal of MaxDif when selecting the data points is designed for MUFs, which may not be optimal for CES functions or AUFs. 
The running times of MaxDif are lower than those of MaxDom as well. It takes 1.7 seconds (82.7\% lower) to process the Random data set (1,068 skyline points) and 92.7 seconds (83.7\% lower)
to process the Anti-correlated data set (12,710 skyline points) when $k = 50$. For the Anti-correlated data set which has a large number of 
skyline points, MinVar (`$\circ$') and AreaGready (`$\triangle$')  offer the most competitive performance. Their maximum regret ratios are close to those of MaxDif, while their running times are within 0.1 seconds. 

\begin{figure}[htp]
\centering
	\subfloat[MRR (Cobb-Douglas)]{
				\hspace{5mm}
		\begin{overpic}[scale=1.06]{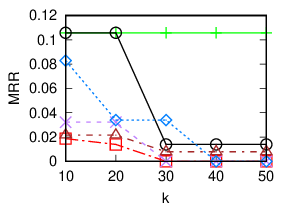}
		\end{overpic}
	}
	\subfloat[MRR (CES)]{
		\begin{overpic}[scale=1.06]{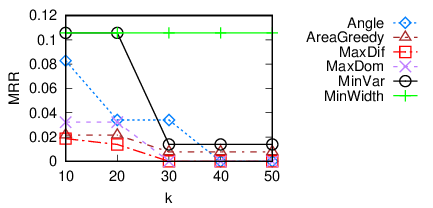}
		\end{overpic}
	}\\
	\vspace{-3mm}
	\subfloat[MRR (AUF)]{
				\hspace{-15mm}
		\begin{overpic}[scale=1.06]{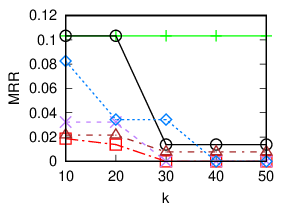}
		\end{overpic}
	}
	\subfloat[Running time]{
		\begin{overpic}[scale=1.06]{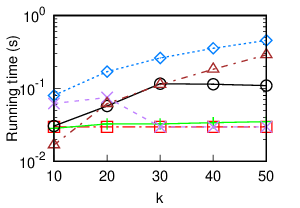}
		\end{overpic}
	}
	\vspace{-1mm}
	\caption{Varying $k$ (Correlated) \label{fig:k_correlated}}
\end{figure}

\begin{figure}[htp]
\centering
	\subfloat[MRR (Cobb-Douglas)]{
			\hspace{5mm}
		\begin{overpic}[scale=1.06]{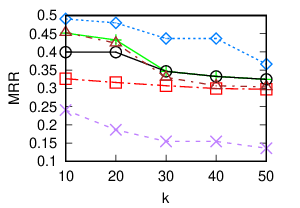}
		\end{overpic}
	}
	\subfloat[MRR (CES)]{
		\begin{overpic}[scale=1.06]{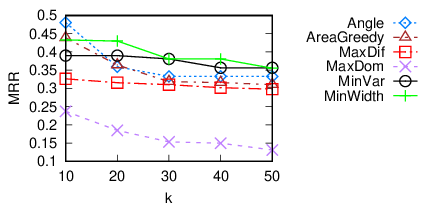}
		\end{overpic}
	}\\
	\vspace{-3mm}
	\subfloat[MRR (AUF)]{
			\hspace{-15mm}
		\begin{overpic}[scale=1.06]{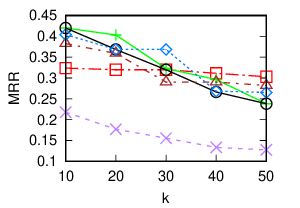}
		\end{overpic}
	}
	\subfloat[Running time]{
		\begin{overpic}[scale=1.06]{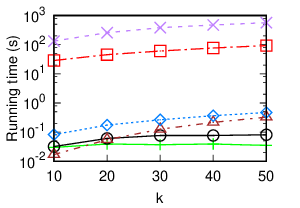}
		\end{overpic}
	}
	\vspace{-1mm}
	\caption{Varying $k$ (Anti-correlated) \label{fig:k_anticorrelated}}
\end{figure}

\begin{figure}[htp]
\centering
	\subfloat[MRR (Cobb-Douglas)]{
		\hspace{5mm}
		\begin{overpic}[scale=1.06]{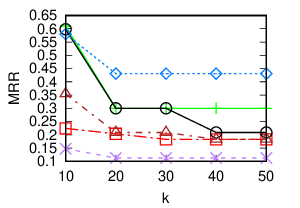}
		\end{overpic}
	}
	\subfloat[MRR (CES)]{
		\begin{overpic}[scale=1.06]{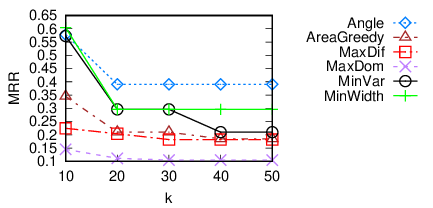}
		\end{overpic}
	}\\
	\vspace{-3mm}
	\subfloat[MRR (AUF)]{
		\hspace{-15mm}
		\begin{overpic}[scale=1.06]{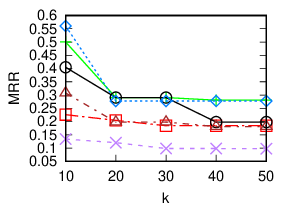}
		\end{overpic}
	}
	\subfloat[Running time]{
		\begin{overpic}[scale=1.06]{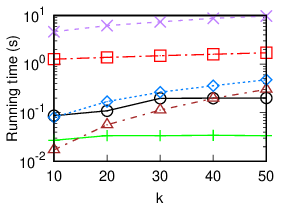}
		\end{overpic}
	}
	\vspace{-1mm}
	\caption{Varying $k$ (Random) \label{fig:k_random}}
\end{figure}

\begin{figure}[htp]
\centering
	\subfloat[MRR (Cobb-Douglas)]{
		\hspace{5mm}
		\begin{overpic}[scale=1.06]{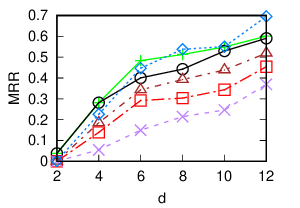}
		\end{overpic}
	}
	\subfloat[MRR (CES)]{
		\begin{overpic}[scale=1.06]{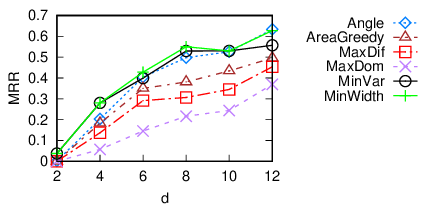}
		\end{overpic}
	}\\
	\vspace{-3mm}
	\subfloat[MRR (AUF)]{
			\hspace{-15mm}
		\begin{overpic}[scale=1.06]{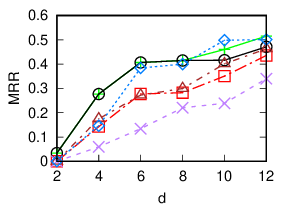}
		\end{overpic}
	}
	\subfloat[Running time]{
		\begin{overpic}[scale=1.06]{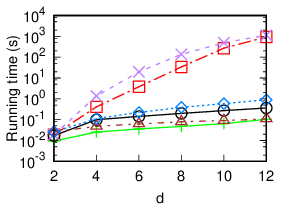}
		\end{overpic}
	}
	\vspace{-1mm}
	\caption{Varying $d$ (Random) \label{fig:d_random}}
\end{figure}

\begin{figure}[htp]
\centering
	\subfloat[MRR (Cobb-Douglas)]{
			\hspace{5mm}
		\begin{overpic}[scale=1.06]{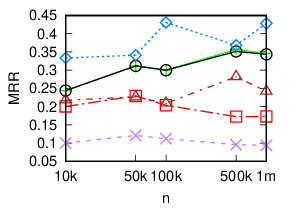}
		\end{overpic}
	}
	\subfloat[MRR (CES)]{
		\begin{overpic}[scale=1.06]{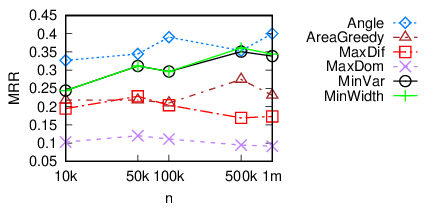}
		\end{overpic}
	}\\
	\vspace{-3mm}
	\subfloat[MRR (AUF)]{
		\hspace{-15mm}
		\begin{overpic}[scale=1.06]{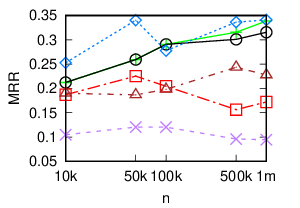}
		\end{overpic}
	}
	\subfloat[Running time]{
		\begin{overpic}[scale=1.06]{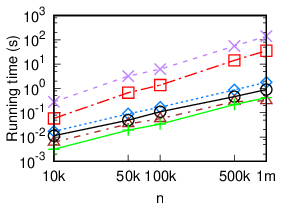}
		\end{overpic}
	}
	\vspace{-1mm}
	\caption{Varying $n$ (Random) \label{fig:n_random}}
\end{figure}

\textbf{Effect of $d$.}
Next, we test the algorithm scalability over the number of data dimensions. 
We vary the number of dimensions $d$ from 2 to 12 with synthetic data. 
Figure~\ref{fig:d_random} shows the result on Random data sets. 
As $d$ increases, the maximum regret ratios increase overall for all the algorithms. This confirms the bounds obtained and is expected, 
since the difference in the utilities of the optimal points in $\mathcal{D}$ 
and $\mathcal{S}$ accumulates when there are more dimensions. 
The increase in the maximum regret ratio is not linear to the increase in $d$, and there are fluctuations, e.g., the maximum regret 
ratios of MinWidth (`$+$') drops when $d$ increases from $8$ to $10$ for the CES functions. 
This is because adding extra dimensions changes the data distribution, the optimal data points, and the data points selected 
into the answer set. Such changes may allow a lower maximum regret ratio for a higher dimensional data set, e.g., a point with large utilities in the extra dimensions is selected 
into the answer set, which compensates the lower utilities in the previous dimensions.  

Similar to Fig.~\ref{fig:k_random}, on Random data sets, MaxDom~(`$\times$')  produces   
the smallest maximum regret ratios, while those of MaxDif (`$\Box$')  are the closest. 
The running times of these two algorithms become too high when $d$ reaches $8$ (33.9 and 137.0 seconds, respectively). For $d \ge 8 $, 
AreaGreedy becomes the most competitive algorithm, as its running time is within 0.1 seconds while its maximum regret ratios 
are close to those of MaxDif.  MinVar (`$\circ$')  has larger maximum regret ratios than those 
of the heuristic algorithms AreaGreedy and (in some cases) Angle (`$\diamond$'), but again 
its  maximum regret ratios do not exceed those of MinWidth (`$+$'). 
We note that the advantage of MinVar over MinWidth in terms of the maximum regret ratio is less significant on this set of experiments with random data, 
as MinVar is designed for handling more skewed data. 

Similar patterns are observed when the algorithms are run on Correlated and Anti-correlated data sets with different numbers of  
dimensions. The comparative performance of the algorithms 
is similar to those shown in Figs.~\ref{fig:k_correlated} and~\ref{fig:k_anticorrelated}.  On Correlated data sets where the numbers of 
skyline points are small, MaxDif and MaxDom are both fast and produce small maximum regret ratios. On Anti-correlated data sets where the numbers of 
skyline points are larger, MinVar and AreaGreeday both have relatively small maximum regret ratios, and they can terminate in a realistic time. We omit the figures 
for conciseness. The same applies for the rest of the experiments.

\textbf{Effect of $n$.}
We further vary the data set cardinality from 10,000 to 1,000,000.
Figure~\ref{fig:n_random} shows the result on Random data sets.
The comparative performance of the algorithms is again similar to that shown in Fig.~\ref{fig:k_random}.
MaxDom (`$\times$') has the smallest maximum regret ratios, while its running time (55.3 seconds) 
is not competitive for data sets with over half a million data points (1,715 skyline points).  
MaxDif (`$\Box$') again has the second smallest maximum regret ratios (except for the AUFs, for reasons discussed earlier), 
while it can process 1 million data points (2,090 skyline points)
within 35.6 seconds. When the number of data points reaches half a million, AreaGreedy (`$\triangle$') again becomes a more competitive heuristic
algorithm as its maximum regret ratios are close to those of MaxDif while its running times are within the 2-second threshold.  
As for the two algorithms MinVar~(`$\circ$')  and MinWidth (`$+$') with bounded maximum regret ratios, 
MinVar again has consistently non-greater maximum regret ratios, and the running times of both algorithms are 
within the 2-second threshold as well.

We also observe that, while the running times of all the algorithms increase with the data set cardinality, which is natural, the maximum regret ratios do not 
show a significant increase. This is because the number of points that yield large utilities for most utility functions increases with the data set cardinality,  
which leads to a higher probability of selecting and inserting one of such points into the answer set, and hence may produce a smaller maximum regret ratio.

\begin{figure}[htp]
\centering
	\subfloat[MRR (Cobb-Douglas)]{
		\begin{overpic}[scale=1.06]{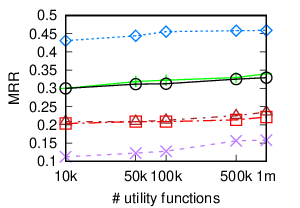}
		\end{overpic}
	}
	\subfloat[MRR (CES)]{
		\begin{overpic}[scale=1.06]{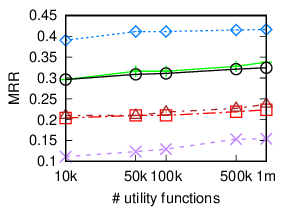}
		\end{overpic}
	}\\
	\vspace{-3mm}
	\subfloat[MRR (AUF)]{
		\begin{overpic}[scale=1.06]{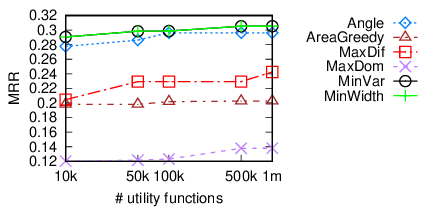}
		\end{overpic}
	}
	\vspace{-1mm}
	\caption{Varying the number of utility functions (Random) \label{fig:q_random}}
\end{figure}

\textbf{Effect of the number of utility functions.}
Figure~\ref{fig:q_random} shows the impact of varying the number of utility functions used in computing the maximum regret ratio 
from 10,000 to 1,000,000 (we generated more utility functions for this set of experiments). Since the algorithms are independent of 
any specific utility function, their output and running times do not change when the utility functions change.  
Only the maximum regret ratios may change. Thus, we  only show the maximum regret ratios and omit the running times for this set of experiments. 

The maximum regret ratios of the algorithms increase with the number of utility functions.
This is because, given the same answer set, when there are more utility functions, 
the probability of satisfying all the utility functions drops, and hence the maximum regret ratio may increase. 
We see from the figure that, while MaxDom (`$\times$')  has the smallest maximum regret ratio, its maximum regret ratio 
increases the most significantly with the number of utility functions. This is because MaxDom selects the set of skyline points 
that dominate the most non-skyline points. It does not consider the difference in the utilities of the selected and unselected skyline points. 
This difference can be large. When there are more utility functions, it is more likely to have an utility function for which 
an unselected skyline point has the largest utility, which can cause a large regret ratio for MaxDom. Our MaxDif  algorithm, on the other hand, 
is more robust to a larger number of utility functions (for the Cobb-Douglas utility functions), as it already considers the entire family of infinite multiplicative utility functions when computing the 
answer set. In particular, for the Cobb-Douglas utility functions, 
the maximum regret ratios of MaxDif (`$\Box$') have been kept steady at around 0.2 and become closer to those of MaxDom 
as more utility functions are used. For the CES and AUF functions, the maximum regret ratios of MaxDif and MaxDom do not become closer as fast, 
because the optimization function of MaxDif may not suit CES or AUF functions.

\begin{figure}[htp]
\centering
	\subfloat[MRR (Cobb-Douglas)]{
		\begin{overpic}[scale=1.06]{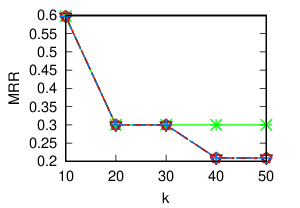}
		\end{overpic}
	}
	\subfloat[MRR (CES)]{
		\begin{overpic}[scale=1.06]{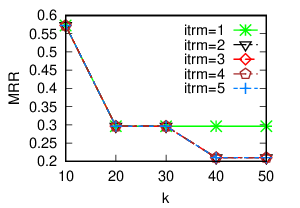}
		\end{overpic}
	}\\
	\vspace{-3mm}
	\subfloat[MRR (AUF)]{
		\begin{overpic}[scale=1.06]{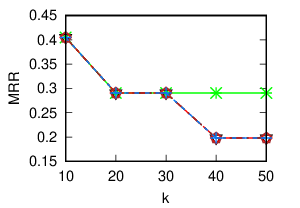}
		\end{overpic}
	}
	\subfloat[Running time]{
		\begin{overpic}[scale=1.06]{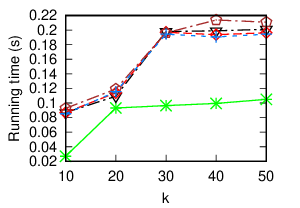}
		\end{overpic}
	}
	\vspace{-1mm}
	\caption{Varying $itr_m$ (Random) \label{fig:itrm_random}}
\end{figure}

\textbf{Effect of $itr_m$.}
Parameter $itr_m$ in MinVar controls the maximum number of iterations to try out different numbers of intervals in each dimension.
In Fig.~\ref{fig:itrm_random},  we run MinVar by 
varying $itr_m$ from 1 to 5 on the Random data set 
to study the impact of $itr_m$.  
We see that the maximum regret ratios decrease as $itr_m$ increases from 
1 to 2 but become stable when $itr_m$ increases further.  
This shows that $t+1$ intervals (i.e., second iteration) is sufficient to 
produce $k$ points to fill up the answer set, because the buckets are mostly non-empty for the Random data set. 
Only two iterations are run even if $itr_m > 2$.
Note that when $k = 20$, the algorithm running times are close for different values of $itr_m$. 
This is because $\displaystyle t = (k-d+1)^\frac{1}{d-1} = 2$ when $k = 20$, which creates  
$t^{d-1} = 16$ buckets and yields 16 points in one iteration. Adding in the four points with the largest coordinates in the first four dimensions, 
a total of 20 points are produced, which are sufficient to answer the query. One iteration is needed even for $itr_m > 1$.
Experiments on the other data sets reveal that $itr_m = 11$ is the largest needed for all data sets tested.  This supports the high efficiency 
of the MinVar algorithm.  

\begin{figure}[htp]
\centering
	\subfloat[MRR (Cobb-Douglas)]{
		\begin{overpic}[scale=1.06]{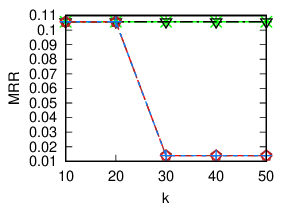}
		\end{overpic}
	}
	\subfloat[MRR (CES)]{
		\begin{overpic}[scale=1.06]{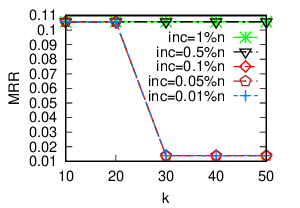}
		\end{overpic}
	}\\
	\vspace{-3mm}
	\subfloat[MRR (AUF)]{
		\begin{overpic}[scale=1.06]{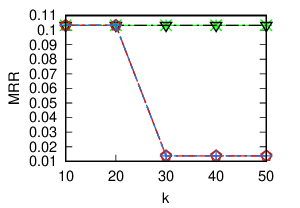}
		\end{overpic}
	}
	\subfloat[Running time]{
		\begin{overpic}[scale=1.06]{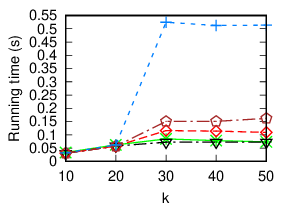}
		\end{overpic}
	}
	\vspace{-2mm}
	\caption{Varying $inc$ (Correlated) \label{fig:inc_correlated}}
	\vspace{-4mm}
\end{figure}

\textbf{Effect of $inc$.}
Parameter $inc$ in MinVar controls the increment step size  
of the sub-algorithm FindBreakPoints to optimize the interval size given $t$. 
We test the impact of $inc$ by varying it from $0.01\%n$ to $1\%n$. 
In Fig.~\ref{fig:inc_correlated}, we present the results on the Correlated data set, 
since the data space partitioning strategy of MinVar targets skewed data where $inc$ has the most impact. 
We see that the maximum regret ratios 
decrease when $inc$ reaches $0.1\%n$. 
This is because the decrease in $inc$ leads to buckets with more similar numbers of data points, which 
 increases the probability of the optimal points being selected.
However, when $inc$ decreases further, the benefit in the maximum regret ratios diminishes, 
while the running time keeps increasing as more iterations are needed.

In general, when $inc = 0.1\%n$, the MinVar algorithm achieves the best balance between the maximum regret ratio and the running time, 
which has been used by the default in the experiments.

\textbf{Discussion.} 
From the experiments, we see that MaxDom often produces small maximum regret ratios. However, \emph{MaxDom 
does not have a bound on the maximum regret ratio. Its maximum regret ratios grow as more utility functions  are used to evaluate the answer set. 
Further, the running time of MaxDom is quadratic over the 
number of data points} when the number of skyline points is close to the number of data points, which soon becomes unrealistic. 
The two proposed algorithms MaxDif and MinVar, on the other hand, have comparable maximum regret ratios, while their running times are lower. 
MinVar should be used if a bounded maximum regret ratio is needed. 
MaxDif does not have a bound but is designed to obtain small maximum regret ratios. It is robust to a large number of utility functions (e.g., 1,000,000).
MaxDif can produce query answers within 2 seconds when the number of skyline points is within 1,000 (for 5-dimensional data). MinVar can produce 
query answers within 1 second when the number of data points is within 1,000,000 (for 5-dimensional data). 
When the data set characteristics are unknown, in order to obtain the best results, a feasible approach is to run 
MaxDif and MinVar at the same time. If MaxDif terminates soon enough (e.g., in 2 seconds), 
its answer is returned. Otherwise, the answer produced by MinVar can be used.

\section{Conclusions}\label{sec:conclusions}

We studied $k$-regret queries with multiplicative utility functions which we showed to be more expressive in modeling 
the diminishing marginal rate of substitution than the additive utility functions studied in prior work~\cite{nanongkai2010regret}.
We presented two  algorithms MinVar and MaxDif for such queries. 
MinVar produces query answers with a bounded maximum regret ratio. 
When applied on $k$-regret queries with Cobb-Douglas functions, MinVar 
achieved a maximum regret ratio not exceeding $\displaystyle O(\ln(1+\frac{1}{k^{\frac{1}{d-1}}}))$;
 when applied on $k$-regret queries with Constant Elasticity of Substitution functions, 
MinVar achieved a  maximum regret ratio not exceeding $\displaystyle O(\frac{1}{k^{\frac{1}{d-1}}})$.
We also showed that, given an infinite set of data points, the maximum regret ratio of any algorithm for $k$-regret queries with multiplicative utility functions has a lower
bound of $\displaystyle \Omega(\frac{1}{k^2})$. 
When the data set is finite, the set of all skyline points 
has a maximum regret ratio of 0, which should be returned as the answer set when 
the number of skyline points does not exceed $k$. 
When there are more than $k$ skyline points, 
MaxDif computes a size-$k$ subset of skyline points to minimize the maximum regret ratio. 
We performed extensive experiments using both real and synthetic data.
The results showed that the regret ratios produced by MinVar and MaxDif are consistently small: 
MaxDif suits the case where the number of skyline points is not too large, 
while MinVar can produce small maximum regret ratios  in real-time as the number of skyline points becomes larger (e.g., over 1,000).

Future work involves exploring the behavior of 
$k$-regret queries with various other types of multiplicative utility functions and parameter settings such as the product of convex (concave) functions and when $\sum_{i=1}^{d}\alpha_i > 1$.
It would  also be interesting to see how $k$-regret queries with a mix of both additive and multiplicative utility functions 
can be answered with a bounded regret ratio.
 
 \begin{acks}
 This work is supported in part by Australian Research Council (ARC)  Discovery Projects DP180102050 and DP180103332 and the National Science
Foundation under Grant IIS-13-20791. We thank Dr.~Ashwin Lall for sharing the code of Angle, AreaGreedy, and MinWidth.
\end{acks}

\bibliographystyle{ACM-Reference-Format}
\small
\bibliography{reference}  

\end{document}